\documentclass{article}
\usepackage{graphicx} 
\usepackage[style=numeric, citestyle=numeric, sorting=none]{biblatex}
\usepackage[left=1in,top=1in,right=1in,bottom=1in]{geometry}
\usepackage{amsfonts, amsmath, amsthm, amssymb}
\usepackage{tabularx, booktabs} 
\usepackage[utf8]{inputenc}
\usepackage{enumitem}
\usepackage{varwidth}
\usepackage{hyperref}
\usepackage{algorithm,algpseudocode}
\usepackage{mathtools}
\makeatletter
\def\plist@algorithm{Alg.\space}
\makeatother
\algdef{SE}[SUBALG]{Indent}{EndIndent}{}{\algorithmicend\ }%
\algtext*{Indent}
\algtext*{EndIndent}
\usepackage{xcolor}
\DeclareNameAlias{author}{last-first}
 
\newcommand{\RR}[0]{\mathbb{R}}

\newcommand{\pl}{\parallel}
\newcommand{\openr}{\hbox{${\rm I\kern-.2em R}$}}
\newcommand{\openn}{\hbox{${\rm I\kern-.2em N}$}}

\newtheorem{theorem}{Theorem}

\newtheorem{lemma}{Lemma}

\title{Adaptive-TMLE for the Average Treatment Effect based on Randomized Controlled Trial Augmented with Real-World Data}
\author{Mark van der Laan$^1$, Sky Qiu$^{1}$\footnote{Corresponding author; E-mail: sky.qiu@berkeley.edu}, Jens Magelund Tarp$^2$, Lars van der Laan$^3$\\
\vspace{0.3cm}\\
$^1$Division of Biostatistics, University of California, Berkeley\\
$^2$Novo Nordisk, Søborg, Denmark\\
$^3$Department of Statistics, University of Washington, Seattle\\}
\date{\today}

\begin{document}
\maketitle
\begin{abstract}
We consider the problem of estimating the average treatment effect (ATE) when both randomized control trial (RCT) data and external real-world data (RWD) are available. We decompose the ATE estimand as the difference between a pooled-ATE estimand that integrates RCT and RWD and a bias estimand that captures the conditional effect of RCT enrollment on the outcome. We introduce an adaptive targeted maximum likelihood estimation (A-TMLE) framework to estimate them. We prove that the A-TMLE estimator is $\sqrt{n}$-consistent and asymptotically normal. Moreover, in finite sample, it achieves the super-efficiency one would obtain had one known the oracle model for the conditional effect of the RCT enrollment on the outcome. Consequently, the smaller and more parsimonious the working model of the bias induced by the RWD is, the greater our estimator's efficiency, while our estimator will always be at least as efficient as an efficient estimator that uses the RCT data only. A-TMLE outperforms existing methods in simulations by having smaller mean-squared-error and 95\% confidence intervals. We apply A-TMLE to augment the DEVOTE trial with external data from the Optum Clinformatics Data Mart, demonstrating its potential to establish treatment superiority in noninferiority trials. A-TMLE could utilize external RWD to help improve the power of randomized trials without biasing the estimates of intervention effects. This approach could allow for smaller, faster trials, decreasing the time until patients can receive effective treatments.
\end{abstract}

\section{Introduction}
The 21st Century Cures Act, enacted by the United States Congress in 2016, was intended to enhance the development process of drugs and medical devices through the use of real-world evidence (RWE), aiming to expedite their availability to patients in need \cite{cures_act}. In response, the U.S. Food and Drug Administration (FDA) released a framework in 2018 for its Real-World Evidence Program, providing comprehensive guidelines on the use of RWE to either support the approval of new drug indications or satisfy post-approval study requirements \cite{fda_rwe}. This framework highlighted the incorporation of external controls, such as data from previous clinical studies or electronic health records, to strengthen the evidence gathered from randomized controlled trials (RCTs), especially when evaluating the safety of an already approved drug for secondary outcomes. Combining RCT data with external real-world data (RWD) would potentially allow for a more efficient and precise estimation of treatment effects, especially if the drug had been approved in specific regions, making data on externally treated patients also available. However, this hybrid design of combining RCT and external RWD, sometimes referred to as \textit{data augmentation}, \textit{data fusion} \cite{bareinboim_causal_2016}, or \textit{data integration}, faces challenges, including differences in the trial population versus the real-world population, which may violate the positivity assumption. Researchers often employ matching techniques to balance patient characteristics between the trial and real-world \cite{zhang_estimation_2021,lin_matching_2023}. Furthermore, existing methods often rely on the assumption of mean exchangeability over the studies, as described in \cite{rudolph_robust_2017,dahabreh_generalizing_2019}, which states that enrolling in the RCT does not affect patients' mean potential outcomes, given observed baseline characteristics. However, this assumption may not hold in reality for many reasons. For example, regional differences in standard-of-care might affect the comparability of the control arm of the RCT and the RWD. If being in the RCT setting increases patients' adherence to their assigned treatment regimes, then the treatment arms also might not be comparable. Additional bias might be introduced from non-concurrency between the RCT and the external RWD, inconsistencies in the outcome measurement, and unmeasured confounding in the non-randomized external study.

In the absence of additional assumptions, investigators are limited to using data from the RCT alone, completely ignoring the external data sources. That is, from an asymptotic perspective, an efficient estimator can totally ignore the external data. However, efficient estimators could still utilize the external data to obtain useful scores from the covariates as dimension reductions to be used in the primary study, which might improve their finite sample performance. For instance, the PROCOVA method learns a prognostic score from historical control arm data and adjusts for this score in the outcome regression to improve the efficiency of estimators in analyzing the current trial \cite{schuler_increasing_2022}. On the other hand, if investigators are willing to make assumptions about the impact of being enrolled in the RCT on the outcome of interest, conditioned on treatment and covariates, then the resulting new statistical model would allow for construction of efficient estimators that are significantly more precise than an estimator ignoring the external study data. Specifically, knowledge on the conditional effect of RCT enrollment on the outcome captures the bias induced from pooling the RCT and external data and therefore allows for their integration.

Making unrealistic assumptions will generally cause biased inference, thereby destroying the sole purpose of the RCT as the \textit{gold-standard} and a study that can provide unbiased estimates of the desired causal effect. To address this problem, we propose an estimation framework that data-adaptively learns a statistical working model that approximates the true impact of the study indicator on the outcome (which represents the bias function for the external study) and then constructs an efficient estimator for the working model specific ATE estimand, defined as first projecting the true data distribution onto the working model, and then plugging into the original target parameter mapping. Our method follows precisely the general adaptive targeted maximum likelihood estimation (A-TMLE) proposed in \cite{adml}. Accordingly, we refer to our method as A-TMLE for data integration.

Our proposed estimator has the following advantages. First, A-TMLE has been shown to be asymptotically normal for the original target parameter, as the difference between the projected and the original estimand being second-order. This holds under the same assumptions as a regular TMLE, provided a sensible nonparametric method for model selection is used, such as cross-validation among a set of candidate working models (e.g., highly adaptive lasso). In our particular application of augmenting RCT data with external RWD, since a regular TMLE would always be asymptotically normal by only using the RCT data in the limit, our estimator would fully preserve the robustness of a regular TMLE of the true target parameter due to being equivalent to a regular TMLE asymptotically. Second, A-TMLE has a super-efficient influence function that equals the efficient influence function of the target parameter when assuming an oracle model that is approximated by the data-adaptive working model. A-TMLE could be viewed as a regularized version of TMLE by performing finite sample bias-variance trade-off in an adaptive fashion. In our application, this implies that the simpler and more parsimonious the impact of the RCT enrollment on the outcome as a function of the treatment and baseline covariates, the smaller the working model would be, leading to more efficiency gain. This further suggests that even if the bias is large in magnitude, provided it is a simple function of treatment and baseline covariates, A-TMLE is still expected to yield efficiency gain. This shows a key advantage over many existing data integration estimators, where efficiency gains are often limited by the magnitude of the bias. Additionally, if the true bias is highly complex, it may still be reasonably well-approximated by a much simpler function in finite sample. Therefore, A-TMLE may still offer efficiency gains in these scenarios. Third, A-TMLE corresponds to a well-behaved TMLE of the working model specific ATE estimand, which is itself of interest, beyond that its approximation error with respect to the true target estimand is of second-order.

To summarize our contributions, A-TMLE provides valid nonparametric inference in problems of integrating RCT data and external RWD while fully utilizing the external data. These estimators are not regular with respect to the original statistical model but are regular under the oracle submodel. As sample size grows, the data-adaptively learned working model will eventually approximate the true model, but in finite sample A-TMLE acts as a super-efficient estimator of the true estimand and efficient estimator of the projection estimand. In our particular estimation problem where even efficient estimators could be underpowered, this class of A-TMLE provides a way forward to achieve more efficiency gain in finite sample without sacrificing nonparametric consistency and valid statistical inference. 

\subsection{Related work}
The decision to integrate or discard external data in the analysis of an RCT often hinges on whether pooling introduces bias. A simple ``test-then-pool'' strategy involves first conducting a hypothesis test to determine the comparability between the RCT and RWD before deciding to pool them together or rely solely on the RCT data \cite{viele_use_2014}. The threshold of the test statistic above which the hypothesis test would be rejected could be determined data-adaptively \cite{yang_elastic_2023}. However, methods of this type can be limited by the small sample sizes typical of RCTs, potentially resulting in underpowered hypothesis tests \cite{li_revisit_2020}. Additionally, using the ``test-then-pool'' strategy, either an efficiency gain is realized, or it is not, without gradation. Bayesian dynamic borrowing methods adjust the weight of external data based on the bias it introduces \cite{pocock_combination_1976,ibrahim_power_2000,hobbs_commensurate_2012,schmidli_robust_2014,lin_many_2023}. Similar frequentist approaches data-adaptively choose between data sources by balancing the bias-variance trade-off, often resulting in a weighted combination of the RCT-only and the pooled-ATE estimands \cite{escvtmle,chen_minimax_2021,cheng_adaptive_2021}. In particular, the experiment-selector cross-validated targeted maximum likelihood estimator (ES-CVTMLE) method \cite{escvtmle} also allows the incorporation of a negative control outcome (NCO) in the selection criteria. The gain in efficiency of these methods largely depends on the bias magnitude, with larger biases diminishing efficiency gains. Other methods incorporate a bias correction by initially generating a biased estimate, then learning a bias function, and finally adjusting the biased estimate to achieve an unbiased estimate of the causal effect, which has a similar flavor to the method we are proposing \cite{kallus_removing_2018,wu_integrative_2022,shyr_multi_2023}. However, the key aspect distinguishing our work and theirs is that their techniques rely on conditional independence between potential outcomes and the study indicator $S$ given observed covariates, an assumption that could easily be violated in reality. For instance, subjects might adhere more strictly to assigned treatment regimens had they been enrolled in the RCT, so being in the RCT or not modifies the treatment effects. This can be viewed as a slight difference in the definition of the treatment in the RCT and the RWD. These methods also focus primarily on unmeasured confounding bias, yet other biases, such as differences in outcome measurement between RCT and RWD settings, may also occur. The reliance on this independence assumption also limits their applicability in scenarios involving surrogate outcomes in the RWD, in which case the endpoints in the RCT and RWD are either inconsistently measured or are differently defined. For instance, digital health endpoints measured using different medical devices may introduce batch effects, where device manufacturer discrepancies cause potential outcomes in the RWD to differ from those in the RCT, even for the same patient. We define our bias directly as the difference between the target estimand and the pooled estimand, therefore it covers any kind of bias, including unmeasured confounding bias. Thus, our method can be applied more broadly, even in cases where the outcomes differ in the two studies.

\subsection{Organization of the article}
This article is organized as follows. In Section \ref{sec:formulation} we formally define the estimation problem in terms of the statistical model and target estimand that identifies the desired average treatment effect. In Section \ref{sec:decomposition}, we decompose our target estimand into a difference between a pooled-ATE estimand $\tilde{\Psi}$ and a bias estimand $\Psi^{\#}$. In Sections \ref{sec:psi_tilde} and \ref{sec:psi_pound}, we provide the key ingredients for constructing A-TMLEs of $\tilde{\Psi}$ and $\Psi^{\#}$, respectively. In Section \ref{sec:psi_tilde}, we present a semiparametric regression working model for the pooled-ATE estimand $\tilde{\Psi}$ and define the corresponding projection parameter. The working model is designed for the conditional average treatment effect, with the projection defined as a squared-error projection of the true outcome regression $E_0(Y\mid A,W)$ (without conditioning on the study indicator) onto the working model $\mathcal{M}_{A,w}$. We derive the canonical gradient and the corresponding TMLE of $\tilde{\Psi}_{{\cal M}_{A,w}}$. Similarly, in Section \ref{sec:psi_pound}, we present the semiparametric regression working model for the conditional effect of the study indicator on the outcome and the corresponding projection parameter $\Psi^{\#}_{{\cal M}_{S,w}}$ that replaces the outcome regression by its squared-error projection onto the working model $\mathcal{M}_{S,w}$. This then defines a working model specific target estimand on the original statistical model and thus defines a new estimation problem that would approximate the desired estimation problem if the working model approximates the true data density. If the study indicator has zero impact on the outcome, then the pooled-ATE estimand that just combines the data provides a valid estimator, thereby allowing full integration of the external data into the estimator; If the study indicator has an impact explained by a parametric form (e.g. a linear combination of a finite set of spline basis functions), then the estimand still integrates the data but carries out a bias correction according to this parametric form. As the parametric form becomes nonparametric, the estimand becomes the estimand that ignores the outcome data in the external study when estimating the outcome regression, corresponding with an efficient estimator. We derive the canonical gradient and the corresponding TMLE of the projection parameter $\Psi^{\#}_{{\cal M}_{S,w}}$. In Section \ref{sec:theory} we analyze the A-TMLE analogue to \cite{adml}, to make this article self-contained. The analysis can be applied to both components of the target parameter, thereby also providing an asymptotic linearity theorem for the resulting A-TMLE of the target parameter. Specifically, we prove that the A-TMLEs of $\tilde{\Psi}$ and $\Psi^{\#}$ are $\sqrt{n}$-consistent, asymptotically normal with possible super-efficient variances. Under an asymptotic stability condition for the working model, they are also asymptotically linear with efficient influence functions that equal the efficient influence functions of the limit of the submodel (oracle model). We discuss the implementation of the proposed A-TMLE in Section \ref{sec:implementation}. In Section \ref{sec:simulation}, we carry out simulation studies to evaluate the performance of our proposed A-TMLE against other methods including ES-CVTMLE \cite{escvtmle}, PROCOVA (a covariate-adjustment method) \cite{schuler_increasing_2022}, a regular TMLE for the target estimand, and a TMLE using RCT data alone. In Section \ref{sec:data_app}, we apply our method to estimate the 1-, 1.5-, and 2-year risks of major adverse cardiac events (MACE) among patients taking insulin degludec versus glargine, using the DEVOTE trial augmented with both an external treatment and comparator arm from Optum’s de-identified Clinformatics database. We conclude with a discussion in Section \ref{sec:discussion}.

\section{The Estimation Problem}\label{sec:formulation}
We observe $n$ independent and identically distributed observations of the random variable $O=(S,W,A,Y)\sim P_0\in\mathcal{M}$, where $P_0$ is the true data-generating distribution and $\mathcal{M}$ is the statistical model. $S\in \{0,1\}$ is an indicator of the unit belonging to an RCT (or a well-designed observational study with no unmeasured confounding), where the causal effect of the treatment on the outcome can be identified; $W\in\RR^d$ is a vector of baseline covariates; $A\in \{0,1\}$ is an indicator of being in the treatment arm; $Y\in\RR$ is a clinical outcome of interest. Note that for the external data, we consider two scenarios, one with only external control arm and the other with both external treatment and external control arms. The problem formulation and estimation are generally the same for those two scenarios. We will make additional remarks at places where there are differences. Throughout this article, we focus on the setup where the RCT has both treatment and control arms. Single-arm trials could be a more challenging setup, as additional non-testable assumptions might be required in order to identify the desired causal effect. Table \ref{tab:notations} contains a list of notations we will use throughout the article. 

\begin{table}[]
\centering
\resizebox{\linewidth}{!}{
\begin{tabular}{c|c}
\toprule
Notation & Description \\ \midrule
$S$ & RCT indicator \\
$W$ & Patient baseline characteristics \\
$A$ & Treatment indicator \\
$Y$ & Outcome \\ \midrule
$Q_P(S,W,A)=E_P(Y\mid S,W,A)$ & Outcome regression \\
$\bar{Q}_P(W,A)=E_P(Y\mid W,A)$ & Outcome regression, marginalized over $S$ \\
$\theta_P(W)=E_P(Y\mid W)$ & Outcome regression, marginalized over $S$ and $A$ \\
$g_P(a\mid W)=P(A=a\mid W)$ & Treatment mechanism, marginalized over $S$ \\
$\Pi_P(s\mid W,A)=P(S=s\mid W,A)$ & RCT enrollment mechanism \\
$\bar{\Pi}_P(s\mid W)=P(S=s\mid W)$ & RCT enrollment mechanism, marginalized over $A$ \\
$\tau_{S,P}(W,A)=E_P(Y\mid S=1,W,A)-E_P(Y\mid S=0,W,A)$ & Conditional average RCT-enrollment effect \\
$\tau_{A,P}(W)=E_P(Y\mid W,A=1)-E_P(Y\mid W,A=0)$ & Conditional average treatment effect (CATE) \\ \midrule
$\Psi^F(P_{O,U})=E_W[E(Y_1-Y_0\mid S=1,W)]$ & Covariate-pooled ATE full-data parameter \\
$\Psi^F_2(P_{O,U})=E_W[E(Y_1-Y_0\mid S=1,W)\mid S=1]$ & RCT-only ATE full-data parameter \\
$\tilde{\Psi}(P_0)=E_0[\tau_{A,0}(W)]$ & Pooled-ATE estimand \\
$\tilde{\Psi}_{\mathcal{M}_{A,w}}(P_0)=E_0[\tau_{A,\beta_0}(W)]$ & Pooled-ATE projection estimand \\
$\Psi^\#(P_0)=E_0[\Pi_0(0\mid W,0)\tau_{S,0}(W,0)-\Pi_0(0\mid W,1)\tau_{S,0}(W,1)]$ & Bias estimand \\
$\Psi^\#_{\mathcal{M}_{S,w}}(P_0)=E_0[\Pi_0(0\mid W,0)\tau_{S,\beta_0}(W,0)-\Pi_0(0\mid W,1)\tau_{S,\beta_0}(W,1)]$ & Bias projection estimand \\
\bottomrule
\end{tabular}
}
\caption{Notations and their descriptions.}
\label{tab:notations}
\end{table}

\subsection{Structural causal model}
We assume a structural causal model (SCM): $S=f_S(U_S)$; $W=f_W(S,U_W)$; $A=f_A(S,W,U_A)$; $Y=f_Y(S,W,A,U_Y)$, where $U=(U_S,U_W,U_A,U_Y)$ is a vector of exogenous errors. The joint distribution of $(U,O)$ is parametrized by a vector of functions $f=(f_S,f_W,f_A,f_Y)$ and the error distribution $P_U.$ The SCM ${\cal M}^F$ is defined by assumptions on these functions and the error distribution. Here ${\cal M}^F$ denotes the set of full data distributions $P_{U,O}$ that satisfy these assumptions on $f$ and $P_U$. The SCM allows us to define the potential outcomes $Y_1=f_Y(S,W,A=1,U_Y)$ and $Y_0=f_Y(S,W,A=0,U_Y)$ by intervening on the treatment node $A$.

\subsection{Statistical model}
We make the following three assumptions:
\begin{enumerate}[label=\textbf{A\arabic*}]
\item \label{A1} $(Y_0,Y_1)\perp A\mid S=1,W$ (randomization in the RCT); 
\item \label{A2} $0<P(A=1\mid S=1,W)<1, P_W\text{-a.e.}$ (positivity of treatment assignment in the RCT);
\item \label{A3} $P(S=1\mid W)>0,P_W\text{-a.e.}$ (positivity of RCT enrollment in the pooled population).
\end{enumerate}
Note that assumptions \ref{A1} and \ref{A2} are satisfied in an RCT (or a well-designed observational study without unmeasured confounding). Importantly, we do not make the assumption that $E(Y_a\mid S=0,W)=E(Y_a\mid W)$, sometimes referred to as the ``mean exchangeability over $S$'' \cite{dahabreh_generalizing_2019}. In other words, for the combined study, the treatment might not be conditionally independent of the potential outcomes given baseline covariates, due to the potential bias introduced by the external data, such as unmeasured confounding. Specifically, we are concerned that $E_WE(Y_1-Y_0\mid W)\not = E_WE(Y\mid W,A=1)-E_WE(Y\mid W,A=0)$, due to $E_WE(Y_1-Y_0\mid W,S=0)\not =E_WE(Y\mid S=0,W,A=1)-E_WE(Y\mid S=0,W,A=0)$. Beyond these three assumptions, we have additional knowledge on $P_0(A=1\mid S,W).$ In particular, if $S=1$ corresponds to an RCT, then $P_0(A=1\mid S=1,W)$ would just be the randomization probability; If the $S=1$-study is not an RCT, we might still be able to make a conditional independence assumption $P_0(A=1\mid S=1,W)=P_0(A=1\mid S=1,W=W_1)$ for some subset $W_1$ of $W$. In some cases, we might also be able to assume that $S$ is independent of $W$, i.e. $P_{W\mid S=1}=P_{W\mid S=0}.$ This independence could be achieved by sampling the external subjects from the same target population as in the $S=1$-study. As we will see later in the identification step, for our target parameter it is crucial that the support of $P_{W\mid S=0}$ is included in the support of $P_{W\mid S=1}$ so that assumption \ref{A3} holds. In this article we will not make assumptions on the joint distribution of $(S,W)$ beyond assumption \ref{A3}, but our results are not hard to generalize to the case that we assume $W\perp S$. In the important case that we augment an RCT with external controls only we have that $P_0(A=1\mid S=0,W)=0$, so that the only purpose of the $S=0$-study is to augment the control arm of the RCT. To briefly summarize, we only rely on assumptions known to hold for a well-designed RCT. The only additional assumption we require is \ref{A3} and could be made plausible in the selection of the external subjects.

Let ${\cal M}$ be the set of possible distributions $P$ of $O$ that satisfies the statistical assumptions \ref{A2} and \ref{A3}. Then, ${\cal M}=\{P_{P_{U,O}}:P_{U,O}\in {\cal M}^F\}$ is the model implied by the full data model ${\cal M}^F$. We can factorize the density of $O$ according to the time-ordering as follows:
$$
p(s,w,a,y)=p_S(s)p_W(w\mid s)p_A(a\mid s,w)p_Y(y\mid s,w,a).
$$ 
Our statistical model only leverages knowledge on $p_A$ and leaves the forms of other factors in the likelihood unspecified.

\subsection{Target causal parameters}
We could consider two candidate target parameters. The first one is the covariate-pooled ATE:
$$
\Psi^F(P_{O,U})=E_WE(Y_1-Y_0\mid S=1,W).
$$
Note that this parameter measures the conditional treatment effect in the RCT, while it takes an average with respect to the RCT and external RWD combined covariate distribution $p_W=p_{W\mid S=1}p_{S=1}+p_{W\mid S=0}p_{S=0}$. Alternatively, we could take the average with respect to the RCT-only covariate distribution $p_{W\mid S=1}$, in which case we have the RCT-only ATE target parameter:
$$
\Psi^F_2(P_{O,U})=E_W[E(Y_1-Y_0\mid S=1,W)\mid S=1].
$$
In the following subsections, we will discuss their identification results and compare their efficient influence functions.

\subsection{Statistical estimand}
Under assumptions \ref{A1}, \ref{A2} and \ref{A3}, the full-data target causal parameter $\Psi^F$ is identified by the following statistical target parameter $\Psi:{\cal M}\rightarrow\openr$ defined by
\[
\Psi(P_0)= E_0[E_0(Y\mid S=1,W,A=1)-E_0(Y\mid S=1,W,A=0)].\]
Note that this estimand is only well-defined if $\min_{a\in \{0,1\}}P(S=1,W=w,A=a)>0$ for $P_W$-a.e. So, if the $S=0$-study has a covariate distribution with a support not included in the support of $P_{W\mid S=1}$, then the estimand is not well-defined. This explains why assumption \ref{A3} is needed. As we will see in the next subsection, the efficient influence function of $\Psi$ indeed involves inverse weighting by $P(S=1\mid W)$. 
Similarly, $\Psi^F_2$ is identified by
\[
\Psi_2(P_0)= E_0\{[E_0(Y\mid S=1,W,A=1)-E_0(Y\mid S=1,W,A=0)]\mid S=1\}.\]
Thus, $\Psi^F(P_{U,O})=\Psi(P_{P_{U,O}})$ for any $P_{U,O}$ in our statistical model $\mathcal{M}.$ Note that the estimand $\Psi_2(P_0)$ only relies on $\min_{a\in \{0,1\}}P(S=1,W=w,A=a)>0$ for $P_{W\mid S=1}$-a.e, which thus always holds by assumption on the $S=1$-study.

We could construct an efficient estimator of $\Psi:{\cal M}\rightarrow\openr$. However, it may still lack power since it would only be using the $S=1$-observations for learning the conditional treatment effect, and similarly for $\Psi_2$. This point is further discussed in the next subsection by analyzing and comparing the nonparametric efficiency bounds of $\Psi$ and $\Psi_2.$ Therefore, we instead pursue a finite sample super-efficient estimator through adaptive-TMLE whose gain in efficiency is adapted to the complexity of the unknown (but learnable) $P_0$.

\subsection{Efficient influence functions of the target causal parameters}
For completeness, we will present canonical gradients $D_{\Psi,P}$ of $\Psi:{\cal M}\rightarrow\openr$ and $D_{\Psi_2,P}$ of $\Psi_2:{\cal M}\rightarrow\openr$, even though we will not utilize them in the construction of our A-TMLE.
\begin{lemma}\label{lem:gradient_psi_and_psi_2}
Consider a statistical model ${\cal M}$ for the distribution of $O$ only possibly making assumptions on $g_P(A\mid S,W)$. 
Let 
\begin{align*}
Q_P(S,W,A)&\equiv E_P(Y\mid S,W,A)\\
g_P(A\mid S,W)&\equiv P(A\mid S,W).
\end{align*}
The efficient influence function of $\Psi:{\cal M}\rightarrow\openr$ at $P$ is given by:
$$
D_{\Psi,P}(O)= Q_P(1,W,1)-Q_P(1,W,0)-\Psi(P)+\frac{S}{P(S=1\mid W)}\cdot\frac{2A-1}{g_P(A\mid 1,W)}(Y-Q_P(S,W,A)).
$$
This is also the canonical gradient in the statistical model that assumes additionally that $S$ is independent of $W$. The canonical gradient of $\Psi_2:{\cal M}\rightarrow\openr$ at $P$ is given by:
$$
D_{\Psi_2,P}(O)= \frac{S}{P(S=1)}(Q_P(1,W,1)-Q_P(1,W,0)-\Psi_2(P))+\frac{S}{P(S=1)}\cdot \frac{2A-1}{g_P(A\mid 1,W)}(Y-Q_P(S,W,A)).
$$
\end{lemma}
The proof can be found in Appendix A. To compare the nonparametric efficiency bound of $\Psi(P_0)$ and $\Psi_2(P_0),$ let 
$$
\sigma^2_P(s,w,a)=E_P[(Y-Q_P(S,W,A))^2\mid S=s,W=w,A=a].
$$
We note that 
\begin{align*}
P D^2_{\Psi,P}(O)&=E_P(Q_P(1,W,1)-Q_P(1,W,0)-\Psi(P))^2\\
&+E_P\left[ \frac{1}{P(S=1\mid W)}\left(\frac{\sigma^2_P(1,W,1)}{g_P(1\mid 1,W)}+\frac{\sigma^2_P(1,W,0)}{g_P(0\mid 1,W)}\right)\right].
\end{align*}
Similarly,
\begin{align*}
PD^2_{\Psi_2,P}(O)&=E_P\left[\frac{P(S=1\mid W)}{P^2(S=1)}(Q_P(1,W,1)-Q_P(1,W,0)-\Psi_2(P))^2\right]\\
&+E_P\left[\frac{P(S=1\mid W)}{P^2(S=1)}\left(\frac{\sigma^2_P(1,W,1)}{g_P(1\mid 1,W)}+\frac{\sigma^2_P(1,W,0)}{g_P(0\mid 1,W)}\right)\right].
\end{align*}
Thus, the variance of the $W$-component of $D_{\Psi,P}$ is a factor of $\sim 1/P(S=1)$ smaller than the variance of the $W$-component of $D_{\Psi_2,P}$. However, the variance of the $Y$-component of $D_{\Psi,P}$ involves the factor $1/P(S=1\mid W)$ versus the factor $\sim 1/P(S=1)$ in $D_{\Psi_2,P,Y}$. Therefore, if $S$ is independent of $W$, then it follows that the variance of $D_{\Psi,P}$ is smaller than the variance of $D_{\Psi_2,P}$, due to a significantly smaller variance of its $W$-component, while having identical $Y$-components. On the other hand, if $S$ is highly dependent on $W$, then the inverse weighting $1/P(S=1\mid W)$ could easily cause the variance of the $Y$-component of $D_{\Psi,P}$ to be significantly larger than the variance of the $Y$-component of $D_{\Psi_2,P}$. Generally speaking, especially when $P(S=1\mid W)$ depends on $W$, the variance of the $Y$ component dominates the variance of the $W$ component. Therefore, without controlling the dependence of $S$ and $W$, one could easily have that the variance of $D_{\Psi,P}$ is larger than the variance of $D_{\Psi_2,P}$. If the goal of data integration is to maximize statistical efficiency and $\Psi$ is preferred over $\Psi_2$, it is important to sample external subjects such that $S$ is approximately independent of $W$. On the other hand, if the aim is to enhance the generalizability of trial findings, inducing dependence between $S$ and $W$ may be reasonable, as it aligns the covariate distribution in the pooled population more closely with that of the trial's intended target population. For example, if certain patient subgroups are underrepresented in the trial, increasing the likelihood of sampling these patients in the external cohort may be appropriate to enhance the interpretation of estimates. For this article, given that our data application prioritizes efficiency gains, we will focus our attention on the target causal parameter $\Psi(P)$, which, as we argued above, would be the preferred choice if $P(S=1\mid W)\approx P(S=1)$. However, our results can be easily generalized to $\Psi_2(P)$.

\section{Decomposition of the Target Estimand as a Difference Between the Pooled-ATE Estimand and a Bias Estimand}\label{sec:decomposition}
One could construct an A-TMLE directly for $\Psi(P_0)$. This requires data-adaptively learning a working model for the conditional average treatment effect (CATE) within the RCT and then constructing a TMLE for the projection estimand implied by this working model \cite{adml}. However, in our data augmentation problem context, such an A-TMLE would not utilize the external data for estimating the outcome regression, as the CATE function is conditioned on $S=1$, i.e., the RCT. This motivates us to consider the following alternative approach. First, we use A-TMLE to estimate a target estimand that prioritizes statistical power by fully utilizing the pooled data from the RCT and external RWD, while acknowledging that it may be biased. Next, we apply A-TMLE again to estimate and correct for this bias. In our setting, a potentially more efficient estimand to target is the pooled-ATE estimand, given by:
$$
\tilde{\Psi}(P_0)\equiv E_0[E_0(Y\mid A=1,W)-E_0(Y\mid A=0,W)].
$$
However, as we argued in Section \ref{sec:formulation}, making this estimand equal to $\Psi(P_0)$ would rely on identification assumptions that may not hold in many scenarios. Therefore, instead of assuming $\tilde{\Psi}=\Psi$, we define the bias estimand:
$$
\Psi^{\#}(P_0)\equiv \tilde{\Psi}(P_0)-\Psi(P_0).
$$
In other words, $\Psi^{\#}$ captures all the bias introduced by using $\tilde{\Psi}$ instead of $\Psi$ as the target. Then, it follows that we could write our target estimand as:
$$
\Psi(P_0)=\tilde{\Psi}(P_0)-\Psi^{\#}(P_0).
$$
Now, we can view our target estimand as applying a bias correction $\Psi^{\#}$ to the (potentially biased) pooled-ATE estimand $\tilde{\Psi}$. In order to apply the type of A-TMLE as presented in \cite{adml}, one needs to parameterize the target estimand as an average of conditional effects. Note that, $\tilde{\Psi}(P)$ is already an average of the conditional treatment effect, one could follow the approach outlined in \cite{adml} to construct an A-TMLE for it. We now turn our attention to the bias estimand. Lemma \ref{lem:decomposition} expresses $\Psi^{\#}$ in terms of the conditional average RCT-enrollment effect and the RCT enrollment mechanism.

\begin{lemma}\label{lem:decomposition}
Let $\tau_{S,P}(W,A)$ and $\Pi_P(s\mid W,A)$ be the conditional average RCT-enrollment effect and the trial enrollment mechanism, where
\begin{align*}
\tau_{S,P}(W,A)&\equiv E_P(Y\mid S=1,W,A)-E_P(Y\mid S=0,W,A)\\
\Pi_P(s\mid W,A)&\equiv P(S=s\mid W,A).
\end{align*}
We have 
$$
\Psi^{\#}(P)=E_P[\Pi_P(0\mid W,0) \tau_{S,P}(W,0)-\Pi_P(0\mid W,1) \tau_{S,P}(W,1)].
$$
For the special case that $P(S=0\mid W,A=1)=0$ (i.e., external data has only a control arm, no treatment arm). Then, the bias parameter becomes
\[
\Psi^{\#}(P)=E_P \Pi_P(0\mid W,0)\tau_{S,P}(W,0).\]
\end{lemma}
\begin{proof}
Note that 
\begin{align*}
E_P(Y\mid W,A=1)&=E_P(Y\mid S=1,W,A=1)P(S=1\mid W,A=1)\\
&+E_P(Y\mid S=0,W,A=1)P(S=0\mid W,A=1),
\end{align*}
so that
\begin{align*}
&E_P(Y\mid W,A=1)-E_P(Y\mid S=1,W,A=1)\\
&=-P(S=0\mid W,A=1)[E_P(Y\mid S=1,W,A=1)-E_P(Y\mid S=0,W,A=1)]
\end{align*}
Thus, $E_WE(Y\mid W,A=1)-E_WE(Y\mid S=1,W,A=1)=-E_W \Pi_P(0\mid W,1)\tau_{S,P}(W,1)$.
Similarly,
\begin{align*}
&E_P(Y\mid W,A=0)-E_P(Y\mid S=1,W,A=0)\\
&=-P(S=0\mid W,A=0)[E_P(Y\mid S=1,W,A=0)-E_P(Y\mid S=0,W,A=0)].
\end{align*}
Thus, $E_WE(Y\mid W,A=0)-E_WE(Y\mid S=1,W,A=0)=-E_W\Pi_P(0\mid W,0)\tau_{S,P}(W,0)$.
Therefore, $\Psi^{\#}(P)=\tilde{\Psi}(P)-\Psi(P)=E_W[\Pi_P(0\mid W,0)\tau_{S,P}(W,0)-\Pi_P(0\mid W,1)\tau_{S,P}(W,1)]$.
\end{proof}
Now, we have parameterized the bias estimand as a weighted combination of two conditional effects, $\tau_{S,P}(W,0)$ and $\tau_{S,P}(W,1)$, where the weights are the probabilities of enrolling in the RCT of the two arms. This estimand resembles the ATE estimand and thereby enables a similar type of A-TMLE as for $\tilde{\Psi}$ with the role of CATE replaced by $\tau_{S,P}$.

Under the decomposition of the target estimand, the amount of efficiency gain from A-TMLE is driven by the complexities of $E(Y\mid W,A=1)-E(Y\mid W,A=0)$, the CATE function in the pooled data, and the bias function $E(Y\mid S=1,W,A)-E(Y\mid S=0,W,A)$. If reasonable efforts are made to ensure that the bias $\Psi^\#$ is small, a simple function might suffice to approximate the true (potentially fully nonparametric) bias function well, thereby offering efficiency gains. On the other hand, if the bias is large in magnitude but can be well-approximated by a parsimonious model, such as an intercept-only model, then we would also expect gains in efficiency. Intuitively, the decomposition strategy can be viewed as a general way of leveraging A-TMLE to achieve efficiency gains. The process begins by identifying an ``easy-to-estimate'' target, which, while potentially biased, is straightforward to estimate efficiently using all the available data. In our context, this would naturally be the pooled-ATE estimand, which can be estimated using the pooled data. Next, we parameterize the bias in terms of conditional effects and use A-TMLE to estimate the bias. The choice of the easy-to-estimate target should be guided by the specific problem setting and the nature of the resulting bias function.

\section{Ingredients for Constructing an A-TMLE for the Pooled-ATE Estimand: Working Model, Canonical Gradient and TMLE}\label{sec:psi_tilde}
In this section, we outline the key components and steps to construct an A-TMLE for the pooled-ATE estimand, $\tilde{\Psi}(P_0)$. Since the pooled-ATE is the ATE on the pooled sample (RCT and external data combined), the estimation procedure mirrors that of the ATE example presented in \cite{adml}. However, for completeness, we present it here within the context of our problem.

\subsection{Semiparametric regression working model for the conditional average treatment effect}
Given the fact that $\bar{Q}_0(W,A)\equiv E_0(Y\mid W,A)=\bar{Q}_0(W,0)+A\tau_{A,0}(W)$, a working model $\mathcal{T}_{A,w}$ for $\tau_{A,0}$ corresponds to a semiparametric regression working model $\bar{{\cal Q}}_w=\{\theta+A\tau_{A,\beta}:\beta,\theta\}$ for $\bar{Q}_0$, where $\theta$ is an unspecified function of $W$. This working model $\bar{{\cal Q}}_{w}$ further implies a corresponding working model ${\cal M}_{A,w}=\{P\in {\cal M}: \bar{Q}_P\in \bar{{\cal Q}}_{w} \}\subset {\cal M}$ for the observed data distribution $P_0$. Let $\tau_{A,\beta_P}\in {\cal T}_{A,w}$ be a projection of $\tau_{A,P}$ onto the working model ${\cal T}_{A,w}$ with respect to a specified loss function. This generally corresponds to mapping a $P\in {\cal M}$ into a projection $\Pi_{\mathcal{M}_{A,w}}(P)\in {\cal M}_{A,w}$. We recommend using the squared-error projection of $\bar{Q}_P$ on $\bar{\mathcal{Q}}_{w}$, defined as:
$$
\bar{Q}_{w,P}=\arg\min_{\bar{Q}\in \bar{{\cal Q}}_{w}}P\left(\bar{Q}_P-\bar{Q}\right)^2.
$$
Lemma \ref{lem:psi_tilde_loss} provides some insights and justifications on using this projection.
\begin{lemma}\label{lem:psi_tilde_loss}
Let $\bar{\mathcal{Q}}_w=\{\theta+A\tau_{A,\beta}:\theta,\beta\}$ be the working model for the outcome regression $\bar{Q}_P(W,A)=E_P(Y\mid W,A)$, where $\theta$ is some unspecified function of $W$ and $\tau_{A,\beta}(W)=\sum_{j}\beta(j)\phi_j(W)$, a linear combination of basis functions. Let $\mathcal{T}_{A,w}=\{\sum_j\beta(j)\phi_j(W):\beta\}$ be the working model for the CATE function $\tau_{A,P}(W)=E_P(Y\mid W,A=1)-E_P(Y\mid W,A=0)$. Let $\bar{Q}_{w,P}=\arg\min_{\bar{Q}\in \bar{{\cal Q}}_w}P(\bar{Q}_P-\bar{Q})^2$
be the squared-error projection of $\bar{Q}_P$ onto the working model $\bar{\mathcal{Q}}_w$. We have that $\bar{Q}_{w,P}=E_P(Y\mid W)+(A-g_P(1\mid W))\sum_j\beta_{P}(j)\phi_j(W)$, where
\[
\beta_P=\arg\min_{\beta}E_Pg_P(1-g_P)(1\mid W)\left(\tau_{A,P}(W)-\sum_j \beta(j)\phi_j(W)\right)^2.\]
Thus, $\sum_j\beta_{P}(j)\phi_j$ represents the projection of $\tau_{A,P}$ onto the working model $\mathcal{T}_{A,w}=\{\sum_j \beta(j)\phi_j:\beta\}$ with respect to a weighted $L^2$-norm using weights $g_P(1-g_P)(1\mid W)$.
We also have 
\[
\beta_P=\arg\min_{\beta}E_P\left(\bar{Q}_P-(A-g_P(1\mid W))\sum_j \beta(j)\phi_j(W)\right)^2, \]
and thereby 
\[
\beta_P=\arg\min_{\beta}E_P\left(Y-(A-g_P(1\mid W))\sum_j\beta(j)\phi_j(W)\right)^2.\]
Here, one could replace $Y$ by $Y-E_P(Y\mid W)$. This also shows that
\[
\tau_{A,P}=\arg\min_{\tau\in\mathcal{T}_{A,w}}E_P[Y-E_P(Y\mid W)-(A-g_P(1\mid W))\tau]^2.\]
Therefore, we can use the following loss function for learning $\tau_{A,P}$:
\[
L_{\theta_P,g_P}(\tau)=[Y-\theta_P(W)-(A-g_P(1\mid W))\tau]^2,\]
where this loss function is indexed by the nuisance parameters $\theta_P(W)=E_P(Y\mid W)$ and $g_P(1\mid W)=P(A=1\mid W)$. This loss function is double robust with respect to misspecification of $\theta_P$ and $g_P$ in the sense that $\tau_{A,P}=\arg\min_{\tau\in\mathcal{T}_{A,w}}E_PL_{\theta,g}(\tau)$ if either $g=g_P$ or $\theta=\theta_P$.
\end{lemma}
The proof is provided in Appendix A. Interestingly, according to Lemma \ref{lem:psi_tilde_loss}, the squared-error projection of the outcome regression $\bar{Q}_P$ onto the semi-parametric regression working model $
\bar{\mathcal{Q}}_w$ corresponds to a weighted squared-error projection of the CATE function $\tau_{A,P}$ onto the corresponding working model $\mathcal
{T}_{A,w}$, where the weights are $g_P(1-g_P)(1\mid W)$ \cite{adml}. As demonstrated in the efficient influence function derived in the following subsection, these weights contribute to stabilizing the variance estimate, explaining why we recommended this particular projection.

Since the true CATE function is unknown, Lemma \ref{lem:psi_tilde_loss} also provides a valid and practical loss function for learning CATE. This loss function is commonly referred to as the $R$-loss \cite{nie_quasi_2021}. We could  fit a relaxed highly adaptive lasso (relaxed-HAL) paired with this loss function with nuisance estimators $\theta_n$ for $\theta_0$ and $g_n$ for $g_0$ \cite{vdl_efficient_2023}. Specifically, given a rich linear model $\{\tau_{A,\beta}=\sum_j\beta(j)\phi_j:\beta\}$ with a large set of 0th-order HAL spline basis functions $\{\phi_j:j\}$, we compute the lasso-estimator
\[
\beta_n=\arg\min_{\beta,\pl\beta\pl_1\leq C_n}\sum_{i=1}^n \left(Y_i-\theta_n(W_i)-(A_i-g_n(1\mid W_i))\sum_j\beta(j)\phi_j(W_i)\right)^2,\]
where $C_n$ is an upper bound on the sectional variation norm that equals the $L_1$-norm of the $\beta$-coefficients \cite{benkeser_hal_2016}. One may also use higher-order spline HAL-MLEs as analyzed and presented in detail in \cite{vdl_higher_2023}. The working model for $\tau_{S,P}$ is then given by ${\cal T}_{A,w}=\{\sum_{j,\beta_{n,j}\not =0}\beta(j)\phi_j:\beta\}$, i.e. all linear combinations of the basis functions with non-zero coefficients. Finally, we use ordinary least squares (OLS) to obtain a maximum likelihood estimator (MLE) $\beta_n^\star$ for $\beta$ within the learned working model.

\subsection{Canonical gradient of $\tilde{\Psi}_{\mathcal{M}_{A,w}}$ at $P$}
The working model $\mathcal{T}_{A,w}$ implies a semiparametric working model $\mathcal{M}_{A,w}$ for $P_0$. Let $\tilde{\Psi}_{\mathcal{M}_{A,w}}(P_0)\equiv\tilde{\Psi}(\Pi_{\mathcal{M}_{A,w}}(P_0))$ be the corresponding projection estimand, where $\Pi_{\mathcal{M}_{A,w}}(P_0)$ is the projection of $P_0$ onto the working model $\mathcal{M}_{A,w}$. The next step of A-TMLE involves deriving the canonical gradient (efficient influence function) of the projection parameter. First, we find the canonical gradient of the $\beta_P$-component.
\begin{lemma}\label{lem:psi_tilde_beta_gradient}
Let $\beta_P$ be defined as in Lemma \ref{lem:psi_tilde_loss} on a nonparametric model. 
The canonical gradient of $\beta$ at $P$ is given by
\[
D_{\beta,P}^r=I_P^{-1}(A-g_P(1\mid W)){\bf \phi}(W)\left(Y-\theta_P(W)-(A-g_P(1\mid W))\sum_j\beta_{P}(j)\phi_j(W)\right),\]
where $I_P=E_P g_P(1-g_P)(1\mid W){\bf \phi}{\bf \phi}^{\top}(W)$. 
By Lemma \ref{lem:psi_tilde_loss}, this can also be written as 
\[
D_{\beta,P}^r=I_P^{-1}(A-g_P(1\mid W)){\bf \phi}(W)(Y-\bar{Q}_{w,P}(W,A)).\]
\end{lemma}
The proof can be found in Appendix A. The next lemma presents the canonical gradient of the projection parameter $\tilde{\Psi}_{\mathcal{M}_{A,w}}$ at $P$.
\begin{lemma}\label{lem:psi_tilde_gradient}
Let the projection parameter $\tilde{\Psi}_{{\cal M}_{A,w}}(P)=E_P [\tau_{A,\beta_P}(W)]=E_P (\bar{Q}_{w,P}(W,1)-\bar{Q}_{w,P}(W,0))$, while the nonparametrically defined parameter $\tilde{\Psi}(P)=E_P [\tau_{A,P}(W)]=E_P (\bar{Q}_P(W,1)-\bar{Q}_P(W,0))$, where $\beta_P=\arg\min_{\beta}E_P g_P(1-g_P)(1\mid W)(\tau_{A,P}(W)-\tau_{A,\beta}(W))^2$, and $\tau_{A,\beta}(W)=\sum_j \beta(j)\phi_j(W)$. 
We have that the canonical gradient of $\tilde{\Psi}_{{\cal M}_{A,w}}$ at $P$ is given by:
\[
D_{{\cal M}_{A,w},\tilde{\Psi},P}=\tau_{A,\beta_P}(W)-E_P [\tau_{A,\beta_P}(W)]+\sum_j D^r_{\beta,P,j}E_P \phi_j(W),\]
where
\[
D_{\beta,P}^r=I_P^{-1}(A-g_P(1\mid W)){\bf \phi}(W)(Y-\bar{Q}_{w,P}(W,A)),\]
with $D_{\beta,P,j}^r$ denoting the $j$-th component of $D_{\beta,P}^r$.
\end{lemma}

\subsection{TMLE of the projection parameter $\tilde{\Psi}_{{\cal M}_{A,w}}$}
To construct a TMLE for the pooled-ATE projection parameter, we first obtain nuisance estimators $g_n$ and $\theta_n$ of $g_0$ and $\theta_0$, respectively. If we use a relaxed-HAL for learning the working model for the conditional effect of $A$ on $Y$, then we have $\beta_n^\star=\arg\min_{\beta}P_n L_{\theta_n,g_n}(\tau_{A,\beta})$, which solves the score equation for $\beta$, $P_n D^r_{\beta,\theta_n,g_n,\beta_n^\star}=0$. We estimate the marginal distribution of $W$, i.e., $P_{W,0}$, with the empirical measure of $W$ given by $P_{W,n}$. We have then solved $P_n D_{{\cal M}_{A,w},\tilde{\Psi},\beta_n^\star,\theta_n,g_n,P_{W,n}}=0$. The TMLE is now the plug-in estimator $\tilde{\Psi}_{\mathcal{M}_{A,w}}(P_n^\star)\equiv \tilde{\Psi}_{{\cal M}_{A,w}}(P_{W,n},\beta_n^\star)$. We refer interested readers to Section \ref{sec:implementation} for a more detailed description on the implementation.

\section{Ingredients for Constructing an A-TMLE for the Bias Estimand: Working Model, Canonical Gradient and TMLE}\label{sec:psi_pound}
We now discuss the construction of an A-TMLE for the bias estimand, $\Psi^\#(P_0)=E_0[\Pi_0(0\mid W,0)\tau_{S,0}(W,0)-\Pi_0(0\mid W,1)\tau_{S,0}(W,1)]$. The steps involve first data-adaptively learning a working model for the conditional effect of the study indicator $S$ on the outcome $Y$. That is, a working model $\mathcal{T}_{S,w}=\{\tau_{S,\beta}:\beta\}$ for $\tau_{S,0}$. This working model implies a semiparametric regression working model $\mathcal{M}_{S,w}$ for $P_0$. We then construct an efficient estimator for the corresponding projection estimand $\Psi^\#_{\mathcal{M}_{S,w}}(P_0)\equiv \Psi^{\#}(\Pi_{\mathcal{M}_{S,w}}(P_0))$.

\subsection{Semiparametric regression working model for the conditional average RCT-enrollment effect}
Since $Q_0(S,W,A)=Q_0(0,W,A)+S\tau_{S,0}(W,A)$, a working model $\mathcal{T}_{S,w}$ for $\tau_{S,0}$ corresponds to a semiparametric regression working model ${\cal Q}_w=\{\theta+S\tau_{S,\beta}:\beta,\theta\}$ for $Q_0=E_0(Y\mid S,W,A)$, where $\theta$ is now an unspecified function of $W,A$. For clarity, we reuse the notations $\beta$ and $\theta$, although they differ from those used in the previous section for $\tilde{\Psi}$. The context, whether referring to $\tilde{\Psi}$ or $\Psi^\#$, will make the distinction clear. The working model ${\cal Q}_{w}$ further implies a working model ${\cal M}_{S,w}=\{P\in {\cal M}: Q_P\in {\cal Q}_{w} \}\subset {\cal M}$ for the observed data distribution $P_0$. Let $\tau_{S,\beta_P}\in {\cal T}_{S,w}$ be a projection of $\tau_{S,P}$ onto this working model ${\cal T}_{S,w}$ with respect to a specified loss function. 
We again recommend the squared-error projection of $Q_P$ on $\mathcal{Q}_{w}$:
\[
Q_{w,P}=\arg\min_{Q\in {\cal Q}_{w}}P\left(Q_P-Q\right)^2.\]
Similar to Lemma \ref{lem:psi_tilde_loss}, we have that 
\[
\tau_{S,\beta_P}=\arg\min_{\tau\in \mathcal{T}_{S,w}}E_P \Pi_P(1-\Pi_P)(1\mid A,W)(\tau_{S,P}-\tau)^2.\]
That is, the squared-error projection of the outcome regression $Q_P$ onto the semi-parametric regression working model $\mathcal{Q}_w$ corresponds to a weighted squared-error projection of the conditional average RCT-enrollment effect $\tau_{S,P}$ onto the corresponding working model $\mathcal
{T}_{S,w}$. The weights $\Pi_P(1-\Pi_P)(1\mid W,A)$ stabilize the efficient influence function for $\beta_P$, which explains why the loss function $L_{\bar{Q},\Pi}(\tau)$ for $\tau_{S,P}$ for estimating $\tau_{S,0}$ avoids any inverse weighting by the RCT-enrollment probability, allowing for more stable and robust estimation. Just like for learning the CATE, a valid and practical loss function to learn the conditional average RCT-enrollment effect would be
\[
L_{\bar{Q}_P,\Pi_P}(\tau)\equiv \left( Y-\bar{Q}_P(W,A)-(S-\Pi_P(1\mid W,A))\tau\right)^2,
\] 
We could then use a relaxed-HAL paired with this loss function with nuisance estimators $\bar{Q}_n$ for $\bar{Q}$ and $\Pi_n$ for $\Pi_0$. We summarize the discussion above in Lemma \ref{lem:psi_pound_loss}.
\begin{lemma}\label{lem:psi_pound_loss}
Let ${\cal Q}_w=\{\theta+S\tau_{S,\beta}:\beta,\theta\}$, where $\tau_{S,\beta}(W,A)=\sum_{j}\beta(j)\phi_j(W,A)$.
Let $Q_{w,P}=\arg\min_{Q\in {\cal Q}_w}P(Q_P-Q)^2$, where $Q_P=E_P(Y\mid S,W,A)$. Let $\tau_{S,P}(W,A)=E_P(Y\mid S=1,W,A)-E_P(Y\mid S=0,W,A)$.
We have that $Q_{w,P}=E_P(Y\mid W,A)+(S-\Pi_P(1\mid W,A))\sum_j\beta_{P}(j)\phi_j(W,A)$, where
\[
\beta_P=\arg\min_{\beta}E_P\Pi_P(1-\Pi_P)(1\mid W,A)\left(\tau_{S,P}(W,A)-\sum_j \beta(j)\phi_j(W,A)\right)^2.\]
Thus, $\sum_j\beta_{P}(j)\phi_j$ represents the projection of $\tau_{S,P}(W,A)$ onto the working model $\{\sum_j \beta(j)\phi_j:\beta\}$ with respect to a weighted $L^2$-norm with weights $\Pi_P(1-\Pi_P)(1\mid W,A)$.
We also have 
\[
\beta_P=\arg\min_{\beta}E_P\left(Q_P-(S-\Pi_P(1\mid W,A))\sum_j \beta(j)\phi_j(W,A)\right)^2, \]
and thereby 
\[
\beta_P=\arg\min_{\beta}E_P\left(Y-(S-\Pi_P(1\mid W,A))\sum_j\beta(j)\phi_j(W,A)\right)^2.\]
Here, one could replace $Y$ by $Y-E_P(Y\mid W,A)$ as well. This also shows that
\[
\tau_{S,\beta_P}=\arg\min_{\tau\in\mathcal{T}_{S,w}}E_P (Y-E_P(Y\mid W,A)-(S-\Pi_P(1\mid W,A))\tau)^2.\]
Therefore, we can use the following loss function for learning $\tau_{S,P}$:
\[
L_{\bar{Q}_P,\Pi_P}(\tau)=(Y-\bar{Q}_P(W,A)-(S-\Pi_P(1\mid W,A))\tau)^2,\]
where the loss function is indexed by two nuisance parameters $\bar{Q}_P(W,A)=E_P(Y\mid W,A)$ and $\Pi_P(1\mid W,A)=P(S=1\mid W,A)$. The loss function is double robust with respect to misspecification of $\bar{Q}_P$ and $\Pi_P$ in the sense that $\tau_{S,\beta_P}=\arg\min_{\tau\in\mathcal{T}_{S,w}}E_PL_{\bar{Q},\Pi}(\tau)$ if either $\Pi=\Pi_P$ or $\bar{Q}=\bar{Q}_P$.
\end{lemma}
Given the definition of the semiparametric regression working model for the conditional average RCT-enrollment effect in Lemma \ref{lem:psi_pound_loss}, which implies a semiparametric regression working model $\mathcal{M}_{S,w}$ for $P_0$, we can now define a corresponding projection parameter for $\Psi^{\#}(P)$. Specifically, 
let $\Psi^{\#}_{{\cal M}_{S,w}}:{\cal M}\rightarrow\openr$ be defined as
\[
\Psi^{\#}_{{\cal M}_{S,w}}(P)\equiv E_P [\Pi_P(0\mid W,0) \tau_{S,\beta_P}(W,0)]-E_P[\Pi_P(0\mid W,1) \tau_{S,\beta_P}(W,1)].\]

\subsection{Canonical gradient of $\Psi^{\#}_{{\cal M}_{S,w}}$ at $P$}
Let's first find the canonical gradient of the $\beta_P$-component.
\begin{lemma}\label{lem:psi_pound_beta_gradient}
Let $I_P=E_P \Pi(1-\Pi)(1\mid W,A){\bf \phi}{\bf \phi}^{\top}(W,A)$. 
Let $\beta_P$ be defined as in Lemma \ref{lem:psi_pound_loss} on a nonparametric model. 
The canonical gradient of $\beta$ at $P$ is given by:
\[
D_{\beta,P}=I_P^{-1}(S-\Pi_P(1\mid W,A)){\bf \phi}(W,A)\left(Y-\bar{Q}_P(W,A)-(S-\Pi_P(1\mid W,A))\sum_j\beta_{P}(j)\phi_j(W,A)\right),\]
where $\bar{Q}_P(W,A)=E_P(Y\mid W,A)$, $\Pi_P(1\mid W,A)=P(S=1\mid W,A)$.
This can also be written as 
\[
D_{\beta,P}=I_P^{-1}(S-\Pi_P(1\mid W,A)){\bf \phi}(W,A)(Y-Q_{w,P}(S,W,A)).\]
\end{lemma}
The proof can be found in Appendix A. Note that the weights $\Pi_P(1-\Pi_P)$ help stabilize the variance. We then need to find the canonical gradient of the projection parameter $\Psi^{\#}_{{\cal M}_{S,w}}$ at $P$. 
\begin{lemma}\label{lem:psi_pound_gradient}
The canonical gradient of $\Psi^{\#}_{{\cal M}_{S,w}}$ at $P$ is given by
\[
D^{\#}_{{\cal M}_{S,w},P}=D^{\#}_{{\cal M}_{S,w},P_W,P}+D^{\#}_{{\cal M}_{S,w},\Pi,P}+D^{\#}_{{\cal M}_{S,w},\beta,P},\]
where
\begin{align*}
D^{\#}_{{\cal M}_{S,w},P_W,P}&= \Pi_P(0\mid W,0)\sum_j \beta_{P}(j)\phi_j(W,0)-\Pi_P(0\mid W,1)\sum_j \beta_{P}(j)\phi_j(W,1)-\Psi^{\#}_{{\cal M}_{S,w}}(P)\\
D^{\#}_{{\cal M}_{S,w},\Pi,P}&=\left\{\frac{A}{g_P(1\mid W)}\tau_{S,\beta_P}(W,1)-\frac{1-A}{g_P(0\mid W)}\tau_{S,\beta_P}(W,0)\right\}(S-\Pi_P(1\mid W,A))\\
D^{\#}_{{\cal M}_{S,w},\beta,P}&=\sum_j D_{\beta,P,j} E_P\Pi_P(0\mid W,0)\phi_j(W,0)-\sum_j D_{\beta,P,j} E_P \Pi_P(0\mid W,1)\phi_j(W,1),
\end{align*}
where
\[
D_{\beta,P}=I_P^{-1}(S-\Pi_P(1\mid W,A)){\bf \phi}(W,A)(Y-Q_{w,P}(S,W,A)).\]
\end{lemma}
The proof can be found in Appendix A.

\subsection{TMLE of the projection parameter $\Psi^{\#}_{{\cal M}_{S,w}}$}
A plug-in estimator of the projection parameter $\Psi^{\#}_{{\cal M}_{S,w}}(P_0)$ requires an estimator of $\Pi_0$, $\beta_0$, $P_{W,0}$, and a model selection method, as discussed earlier, so that the semiparametric regression working model ${\cal M}_{S,w}$ is known, including the working model $\{\tau_{S,\beta}:\beta\}$ for $\tau_{S,0}$. We then need to solve the score equation for each component of the canonical gradient of $\Psi^{\#}_{{\cal M}_{S,w}}$. First, for the $P_W$-component, we use the empirical distribution $P_{W,n}$ of $W$ as the estimator of $P_{W,0}$. Therefore, $P_n D^\#_{\mathcal{M}_{S,w},W,\beta,\Pi,P_{W,n}}=0$ for any $\beta,\Pi$. Then, for the $\beta$-component, we need to construct nuisance estimators $\bar{Q}_n$ of $\bar{Q}_0$ and $\Pi_n$ of $\Pi_0$. Then, with these two nuisance estimates, we can compute the MLE over the working model for $\tau_{S,0}$ with respect to our loss function $L_{\bar{Q}_n,\Pi_n}(\tau)$: 
\[
\beta_n^\star=\arg\min_{\beta}P_n L_{\bar{Q}_n,\Pi_n}(\tau_{S,\beta}).\]
We have then solved $P_n D^{\#}_{{\cal M}_{S,w},\beta,\beta_n^\star,\Pi_n}(O)=0$. Finally, for the $\Pi$-component,
note that we have the representation: $D^{\#}_{{\cal M}_{S,w},\Pi,P}(O)=C(g,\beta)(S-\Pi_P(1\mid W,A))$, where we refer to $C(g,\beta)$ as the clever covariate. With an estimator $g_n$ of $g_0$, we can compute a targeted estimator $\Pi_n^\star$ solving $P_n C(g_n,\beta_n^\star)(S-\Pi_n^\star(1\mid W,A))=0$. The desired TMLE is then given by $\Psi^\#_{\mathcal{M}_{S,w}}(P_n^\star)\equiv \Psi^{\#}_{{\cal M}_{S,w}}(P_{W,n},\Pi_n^\star,\beta_n^\star)$.
For inference with respect to the projection parameter $\Psi^{\#}_{{\cal M}_{S,w}}(P_0)$, we use that
\[
\Psi^{\#}_{{\cal M}_{S,w}}(P_n^\star)-\Psi^{\#}_{{\cal M}_{S,w}}(P_0)=P_n D^{\#}_{{\cal M}_{S,w},P_0}+o_P(n^{-1/2}),\]
while for inference for the original target parameter $\Psi^{\#}(P_0)$, we use that $\Psi^{\#}_{{\cal M}_{S,w}}(P_0)-\Psi^{\#}(P_0)=o_P(n^{-1/2})$, under reasonable regularity conditions as discussed in \cite{adml} and also the section that follows. 

\subsection{Corresponding projection parameter of the original parameter $\Psi(P_0)$ and its A-TMLE}
Let $\mathcal{M}_{w}=\mathcal{M}_{A,w}\times \mathcal{M}_{S,w}$, we can now also define the projection parameter of our target parameter $\Psi$ as
$$
\Psi_{\mathcal{M}_{w}}(P)=\tilde{\Psi}(\Pi_{{\cal M}_{A,w}}(P))-\Psi^{\#}(\Pi_{{\cal M}_{S,w}}(P))=\tilde{\Psi}_{\mathcal{M}_{A,w}}(P)-\Psi^{\#}_{{\cal M}_{S,w}}(P).
$$
Note that we are using different working models for $\Psi^{\#}$ and $\tilde{\Psi}$, since $\tilde{\Psi}_{\mathcal{M}_{A,w}}$ depends on $P$ through $P_W$ and $\tau_{A,\beta_P}$, while $\Psi^{\#}_{{\cal M}_{S,w}}$ depends on $P$ through $P_W,\tau_{S,\beta_P}$ and $\Pi_P$. The final A-TMLE estimator is obtained by plugging in the TMLEs for the projection parameters of $\tilde{\Psi}$ and $\Psi^\#$:
$$
\Psi(P_n^\star)=\Psi_{\mathcal{M}_{w}}(P_n^\star)=\tilde{\Psi}_{\mathcal{M}_{A,w}}(P_n^\star)-\Psi^\#_{\mathcal{M}_{S,w}}(P_n^\star).
$$
We also have the option to estimate $\tilde{\Psi}(P_0)$ with the regular TMLE and only use an A-TMLE for $\Psi^{\#}_{\mathcal{M}_{S,w}}(P_0)$.

\section{Asymptotic Super-Efficiency of A-TMLE}\label{sec:theory}
We now examine the theoretical properties of A-TMLE. Specifically, we establish the asymptotic super-efficiency of the A-TMLE $\Psi_{\mathcal{M}_{w}}(P_n^\star)$ as an estimator of $\Psi(P_0)$, under certain conditions. While this analysis closely mirrors the work in \cite{adml}, it is included here for completeness and to make this article self-contained. It follows the proof of asymptotic efficiency for TMLE applied to $\Psi_{\mathcal{M}_{w}}(P_n^\star)-\Psi_{\mathcal{M}_{w}}(P_0)$, but with additional work to 1) deal with the data-dependent efficient influence curve in its resulting expansion and 2) establish that the bias $\Psi_{\mathcal{M}_{w}}(P_0)-\Psi(P_0)$ is second-order. The required conditions are analogous to the empirical process and second-order remainder conditions necessary for establishing asymptotic efficiency of a standard TMLE, with an added condition on the data-adaptive working model $\mathcal{M}_{w}$ with respect to approximating the true data-generating distribution $P_0$.
 
Let $D_{\mathcal{M}_{w},P}$ be the canonical gradient of $\Psi_{\mathcal{M}_{w}}$ at $P$. Let $R_{\mathcal{M}_{w}}(P,P_0)\equiv \Psi_{\mathcal{M}_{w}}(P)-\Psi_{\mathcal{M}_{w}}(P_0)+P_0D_{\mathcal{M}_{w},P}$ be the exact remainder. Let $P_n^\star\in {\cal M}_{w}$ be an estimator of $P_0$ with $P_n D_{{\cal M}_{w},P_n^\star}=o_P(n^{-1/2})$. An MLE over ${\cal M}_{w}$ or a TMLE starting with an initial estimator  $P_n^0\in {\cal M}_{w}$ targeting $\Psi_{{\cal M}_{w}}(P_0)$ satisfies this efficient score equation condition. 
Then, by definition of the exact remainder, we have
\[
\Psi_{{\cal M}_{w}}(P_n^\star)-\Psi_{{\cal M}_{w}}(P_0)=(P_n-P_0)D_{{\cal M}_{w},P_n^\star}+R_{{\cal M}_{w}}(P_n^\star,P_0)+o_P(n^{-1/2}).\]
We assume that ${\cal M}_{w}$ approximates an oracle model ${\cal M}_0$ that contains $P_0$ so that
\[
\Psi_{{\cal M}_{w}}(P_0)-\Psi_{\mathcal{M}_0}(P_0)=o_P(n^{-1/2}),\]
where $\Psi_{{\cal M}_0}(P)=\Psi(\Pi_{{\cal M}_0}(P))$, i.e., the projection of $P$ onto $\mathcal{M}_0$ applied to $\Psi$. As shown in \cite{adml}, this does not represent an unreasonable condition. For example, if $P_{0,n}\equiv \Pi_{\mathcal{M}_w}(P_0)\in {\cal M}_0$ with probability tending to 1, which holds if ${\cal M}_{w}\subset {\cal M}_0$, noting that $P_0 \tilde{D}_{{\cal M}_{w},P_{0,n}}=0$ for any $\tilde{D}_{{\cal M}_{w},P_{0,n}}$ in the tangent space $T_{{\cal M}_{w}}(P_{0,n})$ of the working model at $P_{0,n}$ (due to $P_{0,n}$ being an MLE of $p\rightarrow P_0\log p$ over $P\in {\cal M}_{w}$). Then, we have
\begin{align*}
\Psi_{{\cal M}_{w}}(P_0)-\Psi(P_0)&=\Psi(P_{0,n})-\Psi(P_0)\\
&=\Psi_{{\cal M}_0}(P_{0,n})-\Psi_{{\cal M}_0}(P_0)\\
&=-P_0 D_{{\cal M}_0,P_{0,n}}+R_{{\cal M}_0}(P_{0,n},P_0)\\
&=(P_{0,n}-P_0) D_{{\cal M}_0,P_{0,n}}+R_{{\cal M}_0}(P_{0,n},P_0)\\
&=(P_{0,n}-P_0)\left\{ D_{{\cal M}_0,P_{0,n}}-\tilde{D}_{{\cal M}_{w},P_{0,n}}\right\}+R_{{\cal M}_0}(P_{0,n},P_0).
\end{align*}
The exact remainder $R_{{\cal M}_0}(P_{0,n},P_0)$ is second-order in $p_{0,n}-p_0$, so that, if ${\cal M}_{w}$ approximates $P_0$ at rate $n^{-1/4}$, then this will be $o_P(n^{-1/2})$.
Regarding the leading term, we can select $\tilde{D}_{{\cal M}_{w},P_{0,n}}$ equal to the projection of $D_{{\cal M}_0,P_{0,n}}$ onto $T_{{\cal M}_{w}}(P_{0,n})$ in $L^2(P_{0,n})$. 
Thus, under the condition $\Pi_{{\cal M}_{w}}(P_0)\in {\cal M}_0$, using Cauchy-Schwarz inequality, we have that $\Psi_{{\cal M}_{w}}(P_0)-\Psi(P_0)$ is a second-order difference in two oracle approximation errors $\pl p_{0,n}-p_0\pl_{\mu}\pl D_{{\cal M}_0,P_{0,n}}-\Pi_{T_{{\cal M}_{w}}(P_{0,n})}(D_{{\cal M}_0,P_{0,n}})\pl_{\mu}$, where $\mu$ is a dominating measure of $P_{0,n}$ and $P_0$. 
Moreover, one can weaken this condition by defining $\tilde{P}_{0,n}=\Pi_{{\cal M}_0}(P_{0,n})$, assuming the above second-order remainder with $P_{0,n}$ replaced by $\tilde{P}_{0,n}$ is $o_P(n^{-1/2})$, and assuming that $\Psi(\tilde{P}_{0,n})-\Psi(P_{0,n})=o_P(n^{-1/2})$, thereby only assuming that $P_{0,n}$ is close enough to ${\cal M}_0$ (instead of being an element of ${\cal M}_0$) defined by a distance induced by $\Psi$. 
So under this reasonable approximation condition on ${\cal M}_{w}$ with respect to an oracle model ${\cal M}_0$, we have
\[
\Psi_{{\cal M}_{w}}(P_n^\star)-\Psi(P_0)=(P_n-P_0)D_{{\cal M}_{w},P_n^\star}+R_{{\cal M}_{w}}(P_n^\star,P_0)+o_P(n^{-1/2}).\]
Under the condition that $P_n^\star$ converges to $P_0$ at a fast enough rate (i.e., $o_P(n^{-1/4})$ so that
$R_{{\cal M}_{w}}(P_n^\star,P_0)=o_P(n^{-1/2})$, this yields
\[
\Psi_{{\cal M}_{w}}(P_n^\star)-\Psi(P_0)=(P_n-P_0)D_{{\cal M}_{w},P_n^\star}+o_P(n^{-1/2}).\]
This condition would hold if we use HAL to construct the initial estimator of $P_0$ in the TMLE. 
Under the Donsker class condition that $D_{{\cal M}_{w},P_n^\star}$ falls in  a $P_0$-Donsker class with probability tending to 1, this implies already $\Psi_{{\cal M}_{w}}(P_n^\star)-\Psi(P_0)=o_P(n^{-1/2})$, consistency at the parametric rate $n^{-1/2}$. Inspection of the canonical gradient for our projection parameter $\Psi_{{\cal M}_{w}}$ shows that this is not a stronger condition than the Donsker class condition $D_{\Psi,P_n^\star}$ one would need for the TMLE of $\Psi(P_0)$. The condition on the working model ${\cal M}_{w}$ would hold if the model $\bar{{\cal Q}}_{w}$ (or $\mathcal{Q}_w$, for $\Psi^\#$) falls with probability tending to one in a nice class such as the class of \textit{càdlàg} functions with a universal bound on the sectional variation norm. 

Moreover, under the same Donsker class condition, using the consistency of $P_n^\star$, it also follows that
\[
\Psi_{{\cal M}_{w}}(P_n^\star)-\Psi(P_0)=(P_n-P_0)D_{{\cal M}_{w},P_0}+o_P(n^{-1/2}).\]
Finally, if  $D_{{\cal M}_{w},P_0}$ remains random in the limit but is asymptotically independent of the data $P_n$ so that $n^{1/2}P_n D_{{\cal M}_{w},P_0}/\sigma_{0,n}\Rightarrow_d N(0,1)$ with $\sigma_{0,n}^2\equiv P_0 D_{{\cal M}_{w},P_0}^2$, then we obtain  
\[
\sigma_{0,n}^{-1}n^{1/2}\left( \Psi_{{\cal M}_{w}}(P_n^\star)-\Psi(P_0)\right)\Rightarrow_d N(0,1),\]
thereby allowing for the construction of Wald-type confidence intervals for $\Psi(P_0)$. 
If, in fact, $D_{{\cal M}_{w},P_0}$ converges to a fixed $D_{{\cal M}_0,P_0}$, then we have asymptotic linearity:
\[
\Psi_{{\cal M}_{w}}(P_n^\star)-\Psi(P_0)=P_n D_{{\cal M}_0,P_0}+o_P(n^{-1/2}),\]
with influence curve the efficient influence curve of $\Psi_{{\cal M}_0}:{\cal M}\rightarrow\openr$ at $P_0$. Moreover, \cite{adml} shows that $D_{{\cal M}_0,P_0}$ equals the efficient influence curve of $\Psi:{\cal M}_0\rightarrow \openr$ that a priori assumes the model ${\cal M}_0$, showing that this A-TMLE is super-efficient achieving the efficiency bound for estimation $\Psi(P_0)$ under the oracle model ${\cal M}_0$ (as if we are given $P_0\in {\cal M}_0$ for given ${\cal M}_0$). Thus, if the true $\tau_{A,0}$ and $\tau_{S,0}$ are captured by simple models, then the A-TMLE will achieve a large gain in efficiency relative to the regular TMLE, while if $\tau_{A,0}$ and $\tau_{S,0}$  complex so that ${\cal M}_0$ is close to ${\cal M}$, then the gain in efficiency will be small or the resulting A-TMLE will just be asymptotically efficient. The discussions above are summarized into the following theorem.

\begin{theorem}\label{atmletheorem1}
Let $D_{{\cal M}_{w},P}$ be the canonical gradient of $\Psi_{{\cal M}_{w}}=\Psi(\Pi_{{\cal M}_{w}}(\cdot))$ at $P$. Define the exact remainder as $R_{{\cal M}_{w}}(P,P_0)\equiv \Psi_{{\cal M}_{w}}(P)-\Psi_{{\cal M}_{w}}(P_0)+P_0 D_{{\cal M}_{w},P}$. Let ${\cal M}_0\subset{\cal M}$ denote the oracle model containing $P_0$, with $D_{{\cal M}_0,P}$ and $R_{{\cal M}_0}(P,P_0)$ being the canonical gradient and exact remainder of $\Psi_{{\cal M}_0}:{\cal M}\rightarrow\openr$, respectively. Let $P_n^\star\in {\cal M}_{w}$ be an estimator of $P_0$ such that $P_n D_{{\cal M}_{w},P_n^\star}=o_P(n^{-1/2})$. Then,
\[
\Psi_{{\cal M}_{w}}(P_n^\star)-\Psi_{{\cal M}_{w}}(P_0)=(P_n-P_0)D_{{\cal M}_{w},P_n^\star}+R_{{\cal M}_{w}}(P_n^\star,P_0).\]
{\bf Oracle model approximation condition:}
Let $P_{0,n}\equiv \Pi_{{\cal M}_{w}}(P_0)$ and $\tilde{P}_{0,n}\equiv \Pi_{{\cal M}_0}(P_{0,n})$. 
Let  $\tilde{D}_{{\cal M}_{w},P_{0,n}}$ be the projection of $D_{{\cal M}_0,P_{0,n}}$ onto the tangent space $T_{{\cal M}_{w}}(P_{0,n})\subset L^2_0(P_{0,n})$ of the working model ${\cal M}_{w}$ at $P_{0,n}$. 
Assume that ${\cal M}_{w}$ approximates ${\cal M}_0$ in the sense that $\Psi(\tilde{P}_{0,n})-\Psi(P_{0,n})=o_P(n^{-1/2})$ and 
\[
(\tilde{P}_{0,n}-P_0)\left\{ D_{{\cal M}_0,\tilde{P}_{0,n}}-\tilde{D}_{{\cal M}_{w},\tilde{P}_{0,n}}\right\}=o_P(n^{-1/2}).
\]
Then, we have
\[
\Psi_{{\cal M}_{w}}(P_n^\star)-\Psi(P_0)=(P_n-P_0)D_{{\cal M}_{w},P_n^\star}+R_{{\cal M}_{w}}(P_n^\star,P_0)+o_P(n^{-1/2}).\]
{\bf Rate of convergence condition:}
If $P_n^\star$ converges to $P_0$ at a fast enough rate (i.e., $o_P(n^{-1/4})$ so that
$R_{{\cal M}_{w}}(P_n^\star,P_0)=o_P(n^{-1/2})$, then we have
\[
\Psi_{{\cal M}_{w}}(P_n^\star)-\Psi(P_0)=(P_n-P_0)D_{{\cal M}_{w},P_n^\star}+o_P(n^{-1/2}).\]
{\bf Donsker class condition:}
Under the Donsker class condition that $D_{{\cal M}_{w},P_n^\star}$ falls in  a $P_0$-Donsker class with probability tending to 1, it follows that $\Psi_{{\cal M}_{w}}(P_n^\star)-\Psi(P_0)=o_P(n^{-1/2})$, consistency at the parametric rate $n^{-1/2}$. Moreover, under the same Donsker class condition, using the consistency 
$P_0\{D_{{\cal M}_{w},P_n^\star}-D_{{\cal M}_{w},P_0}\}^2\rightarrow_p 0$ of $P_n^\star$, it also follows that
\[
\Psi_{{\cal M}_{w}}(P_n^\star)-\Psi(P_0)=(P_n-P_0)D_{{\cal M}_{w},P_0}+o_P(n^{-1/2}).\]
{\bf Asymptotic normality condition:}
Finally, if  $D_{{\cal M}_{w},P_0}$ remains random in the limit but is asymptotically independent of the data $P_n$ in the sense that $n^{1/2}P_n D_{{\cal M}_{w},P_0}/\sigma_{0,n}\Rightarrow_d N(0,1)$ with $\sigma_{0,n}^2\equiv P_0 D_{{\cal M}_{w},P_0}^2$, then we obtain  
\[
\sigma_{0,n}^{-1}n^{1/2}\left( \Psi_{{\cal M}_{w}}(P_n^\star)-\Psi(P_0)\right)\Rightarrow_d N(0,1),\]
thereby allowing for the construction of Wald-type confidence intervals for $\Psi(P_0)$. 
If, in fact,   $D_{{\cal M}_{w},P_0}$ converges to a fixed $D_{{\cal M}_0,P_0}$, then we have asymptotic linearity:
\[
\Psi_{{\cal M}_{w}}(P_n^\star)-\Psi(P_0)=P_n D_{{\cal M}_0,P_0}+o_P(n^{-1/2}),\]
with influence curve the efficient influence curve of $\Psi_{{\cal M}_0}:{\cal M}\rightarrow\openr$ at $P_0$. Moreover, $D_{{\cal M}_0,P_0}$ equals the efficient influence curve of $\Psi:{\cal M}_0\rightarrow \openr$ that a priori assumes the model ${\cal M}_0$.
\end{theorem}

\section{Implementation of A-TMLE}\label{sec:implementation}
In this section, we describe the implementation of A-TMLE for data augmentation and discuss practical strategies for data-adaptively learning the working models. Broadly, the objective is to construct two estimators, one for the pooled-ATE estimand and the other for the bias estimand. The pooled-ATE is estimated by combining the RCT data with the external RWD. However, direct pooling is likely to introduce bias, requiring a subsequent step to estimate and correct this bias. The bias can be expressed as a weighted combination of the conditional average RCT-enrollment effects for the treatment and control arms. Estimating the bias estimand parallels the estimation of an ATE, with trial enrollment acting as the treatment variable. The final estimate of the target parameter is then obtained by subtracting the bias estimate from the pooled-ATE estimate. The complete procedure for implementing the A-TMLE algorithm is detailed in Algorithm \ref{alg:atmle}.
\begin{algorithm}
\caption{Adaptive-TMLE}
\label{alg:atmle}
\begin{algorithmic}[1]
\State Construct initial estimator $\tilde{\Psi}_{\mathcal{M}_{A,w}}(P_n)$ of $\tilde{\Psi}_{\mathcal{M}_{A,w}}(P_0)$:
\Indent
\State Compute nuisance estimate $g_n$ of $g_0(1\mid W)\equiv P_0(A=1\mid W)$;
\State Compute nuisance estimate $\theta_n$ of $\theta_0(W)\equiv E(Y\mid W)$;
\State Define pseudo outcome $\tilde{Y} \equiv (Y-\theta_n)/(A-g_n)$ and weight $\tilde{\mathcal{W}}\equiv (A-g_n)^2$;
\State Fit a relaxed-HAL using $W$ as covariates, $\tilde{Y}$ as outcome, and $\tilde{\mathcal{W}}$ as weight;
\State Get initial estimate $\tilde{\Psi}_{\mathcal{M}_{A,w}}(P_n)= 1/n\sum_{i=1}^n\tau_{A,\beta_n^\star}(W_i)= \tilde{\Psi}_{\mathcal{M}_{A,w}}(P_n^\star)$;
\EndIndent
\State Construct initial estimator $\Psi_{\mathcal{M}_{S,w}}^\#(P_n)$ of $\Psi_{\mathcal{M}_{S,w}}^\#(P_0)$:
\Indent
\State Obtain nuisance estimate $\Pi_n$ of $\Pi_0(1\mid W,A)\equiv P(S=1\mid W,A)$;
\State Obtain nuisance estimate $\bar{Q}_n$ of $\bar{Q}_0(W,A)\equiv E(Y\mid W,A)$;
\State Compute pseudo outcome $Y^\#\equiv (Y-\bar{Q}_n)/(S-\Pi_n)$ and weight $\mathcal{W}^\#\equiv (S-\Pi_n)^2$;
\State Fit a relaxed-HAL using $(W,A)$ as covariates, $Y^\#$ as outcome, and $\mathcal{W}^\#$ as weight;
\State Get initial estimate $\Psi_{\mathcal{M}_{S,w}}^\#(P_n)= 1/n\sum_{i=1}^n[\Pi_n(0\mid W_i,0)\tau_{S,\beta_n^\star}(W_i,0)-\Pi_n(0\mid W_i,1)\tau_{S,\beta_n^\star}(W_i,1)]$;
\EndIndent
\State Additional targeting for $\Pi$-component:
\Indent
\State Compute clever covariate $C(g_n,\beta_n^\star)=A/g_n(1\mid W)\tau_{S,\beta_n^\star}(W,1)-(1-A)/g_n(0\mid W)\tau_{S,\beta_n^\star}(W,0)$; 
\State Use $C(g_n,\beta_n^\star)$ to perform a TMLE targeting and obtain an updated estimate $\Pi^\star_n$ of $\Pi_0$;
\State Get TMLE $\Psi_{\mathcal{M}_{S,w}}^{\#}(P_n^\star)\equiv 1/n\sum_{i=1}^n[\Pi_n^\star(0\mid W_i,0)\tau_{S,\beta_n^\star}(W_i,0)-\Pi_n^\star(0\mid W_i,1)\tau_{S,\beta_n^\star}(W_i,1)]$;
\EndIndent
\State Compute the final estimate $\Psi_{\mathcal{M}_{w}}(P_n^\star)=\tilde{\Psi}_{\mathcal{M}_{A,w}}(P_n^\star)-\Psi^{\#}_{{\cal M}_{S,w}}(P_n^\star)$;
\State Compute the 95\% confidence interval given by $\Psi_{\mathcal{M}_{w}}(P_n^\star)\pm 1.96\sqrt{P_n(D_{\mathcal{M}_{A,w},\tilde{\Psi}}-D_{\mathcal{M}_{S,w},\Psi^\#})^2/n}$.
\end{algorithmic}
\end{algorithm}

We make the following remarks regarding the algorithm. To estimate nuisance parameters, one might use either highly adaptive lasso (HAL) \cite{benkeser_hal_2016,hejazi_hal_2020}, or super learner \cite{vdl_sl_2007}. Super learner, in particular, allows the integration of a rich library of flexible machine learning algorithms. For a practical guidance on specifying a super learner, we recommend \cite{phillips_practical_2023}. For learning the two working models, we suggest using relaxed-HAL with the loss functions introduced in Sections \ref{sec:psi_tilde} and \ref{sec:psi_pound}. The pseudo-outcome and weight transformations defined in Algorithm \ref{alg:atmle} result from rewriting the loss function. Relaxed-HAL involves initially fitting a HAL where the $L_1$-norm of the coefficients is penalized. Cross-validation is used to select the optimal $L_1$-norm. After that, an ordinary least squares (OLS) regression is applied to the spline basis functions with non-zero coefficients. Because relaxed-HAL is a maximum likelihood estimator (MLE), the empirical mean of the $\beta$-component of the canonical gradient is automatically solved. This eliminates the need to further target the initial estimator $\tilde{\Psi}_{\mathcal{M}_{A,w}}$. Notably, HAL with 0th-order spline basis functions converges in loss-based dissimilarity at a rate of $n^{-1/3}(\log n)^{2(d-1)/3}$ \cite{bibaut_fast_2019}, only assuming that the target functions are \textit{càdlàg} (right-continuous with left limits) and have finite sectional variation norm. Importantly, this rate depends on the dimensionality $d$ only via a power of the $\log n$-factor thus is essentially dimension-free. This allows the learned working model to approximate the truth well, enabling valid inference. Moreover, if the target functions are smooth, one can use $k$-th order spline HAL-MLEs that converge at even better rates $n^{-k^\star/(2k^\star+1)},k^\star=k+1$, up till a $\log n$-factor: for example, the first-order spline HAL-MLE converges at the dimension-free rate $n^{-2/5}$ up till a $\log n$-factor.

In Appendix D, we also discussed several alternative approaches for constructing data-adaptive working models, including learning the working model on independent data sets, and using deep learning architectures to generate dimension reductions of data sets. To complement the methods discussed in Appendix D, we also proposed and analyzed a cross-validated version of A-TMLE in Appendix C to address concerns regarding potential overfitting of the working models as a result of being too data-adaptive. In addition, as suggested in the Introduction, opting for a slightly larger model than indicated by cross-validation could reduce bias. Techniques like undersmoothing as described in \cite{vdl_efficient_2023} are worth considering. Furthermore, to ensure robustness in model selection, one might enforce a minimal working model based on domain knowledge, specifying essential covariates that must be included in the working model. In the context of HAL, this would mean unpenalizing certain covariates or basis functions. This option is available in standard lasso software like the `glmnet' R package \cite{friedman_regularization_2010}.

To estimate the nuisance parameters, one can use a super learner with a library of flexible machine learning algorithms. A discrete super learner, in particular, acts as an ``oracle selector,'' meaning that asymptotically it selects the algorithm that performs best within the library. This property holds even when the number of learners grows polynomially with the sample size, enabling the inclusion of a rich and diverse library of learners. If the super learner library includes HAL as a candidate learner, using it to estimate the nuisance parameters guarantees efficiency due to HAL's fast rate of convergence. We note that proving efficiency also requires satisfying a Donsker class condition. This condition can often be bypassed by using cross-fitting. However, since HAL is an empirical risk minimizer within a cádlág function class, which is known to be Donsker, this condition is automatically satisfied when using HAL. In our simulations and real-data application, we applied cross-fitting when estimating the nuisance parameters. Intuitively, cross-fitting the nuisance parameters prevents overfitting, ensuring sufficient signal remains to accurately learn the working models for the conditional average treatment effect and the conditional average RCT-enrollment effect.

\section{Simulation Studies}\label{sec:simulation}
An ideal estimator in our problem setting should have several important operating characteristics including valid type I error control, nominal confidence interval coverage, and finite sample efficiency gain (compared to using RCT data only). We conducted two sets of simulations to evaluate the performance of A-TMLE in these aspects. These simulations covered a range of scenarios, including simple to complex bias functions, varying sample sizes of external data relative to the RCT, and practical violations of trial enrollment positivity. Descriptions of the five estimators we assessed are shown in Table \ref{tab:estimators}.
\begin{table}[h]
\centering
\resizebox{12cm}{!}{
\begin{tabular}{>{\raggedright\arraybackslash}p{3cm}>{\raggedright\arraybackslash}p{9cm}}
\toprule
\textbf{Estimator} & \textbf{Description} \\ \midrule
A-TMLE & Our proposed estimator, applying the A-TMLE framework to estimate both the pooled-ATE $\tilde{\Psi}$ and the bias $\Psi^{\#}$. \\ \midrule
ES-CVTMLE & An estimator for integrating RCT with RWD within the TMLE framework, data-adaptively chooses between RCT-only or pooled-ATE to optimize bias-variance trade-off \cite{escvtmle}.\\ \midrule
PROCOVA & A prognostic score covariate-adjustment method \cite{schuler_increasing_2022}. \\ \midrule
TMLE & A standard TMLE for the target parameter. \\ \midrule
RCT-only & A standard TMLE for ATE parameter using RCT data alone, serving as the best estimator in the absence of external data.\\
\bottomrule
\end{tabular}
}
\caption{Estimators assessed in simulation studies and their descriptions.}
\label{tab:estimators}
\end{table}
Our primary metrics for evaluating the performance of the estimators include mean-squared-error (MSE), relative MSE, coverage and average width of 95\% confidence intervals. All data generating processes are detailed in Appendix B.

Figure \ref{fig:simple} presents results for scenarios with simple bias functions and varying external data sample sizes relative to the RCT. In scenarios (a) and (b), the bias functions (i.e., the conditional effect of the study indicator $S$ on the outcome $Y$) have simple main-term linear forms. Directly pooling the two data sources would yield a biased estimate of the true ATE due to the presence of the bias function. In the RCT, the probability of assignment to the treatment group is 0.67, which is common in settings where a drug has been approved and the focus is on assessing safety or secondary endpoints (also known as ``asymmetric allocation''). The external dataset comprises both treatment and control arms patients. Treatment assignment in the external data is determined by baseline patient characteristics, as would be the case in routine clinical practice. To reflect different stages of a drug’s post-market lifecycle and varying availability of external real-world data, we vary the external data sample sizes from one to three times that of the RCT.
\begin{figure}[h!]
\centering
\includegraphics[width=0.75\textwidth]{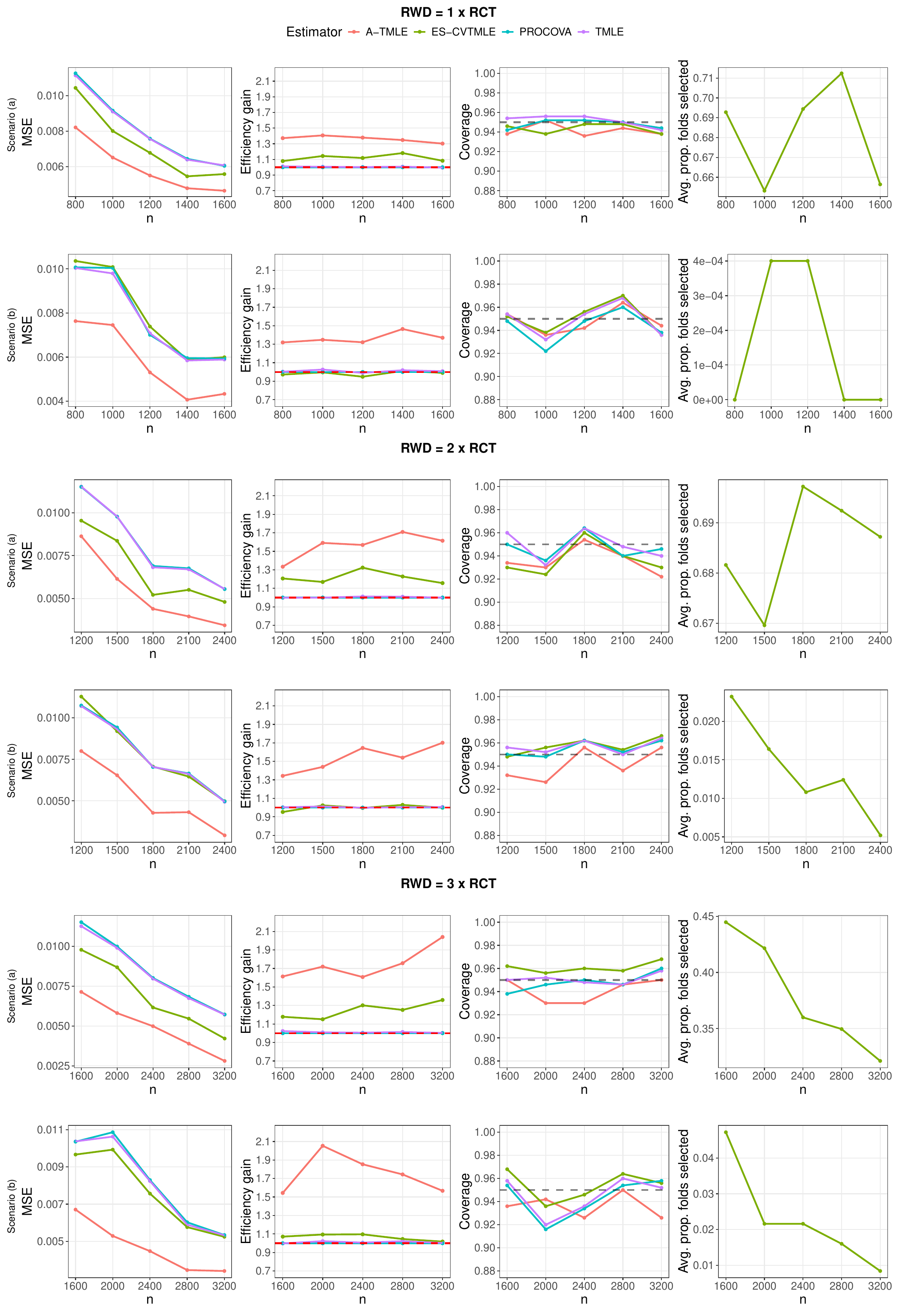}
\caption{Comparison of MSE, efficiency gain, and 95\% confidence interval coverage for A-TMLE, ES-CVTMLE, PROCOVA, and TMLE estimators across increasing sample sizes of both the RCT and external data. Efficiency gain is calculated as the ratio of the MSE of the reference estimator (red dashed line, representing the RCT-only estimator) to the MSE of other estimators. The plots in the right-most column show the average proportion of cross-validation folds in which the ES-CVTMLE estimator pooled the external data, averaged across Monte Carlo runs for each sample size. Scenarios (a) and (b) represent two different data-generating-processes involving main-term parametric models for the conditional average RCT-enrollment effect.}
\label{fig:simple}
\end{figure}
We observe that the A-TMLE estimator consistently achieves the lowest MSE compared to alternative estimators. As the external data sample size increases from one to three times that of the RCT, the efficiency gain relative to a regular TMLE applied to the RCT data alone grows from approximately 1.1 to 2 times. All estimators show roughly 95\% confidence interval coverage. The ES-CVTMLE method partitions the data into $V$ folds. For each fold, it estimates the MSE on the larger $(V-1)$-fold training set, then data-adaptively selects either the $\tilde{\Psi}$ or $\Psi$ based on the magnitude of MSE and estimates the chosen target parameter on the held-out fold. The final estimate is obtained as the average across folds. For ES-CVTMLE, we plotted the average proportion of folds that pooled the external data across Monte Carlo runs in the last column of Figure \ref{fig:simple}. 

In scenario (a), the bias introduced by pooling external data is relatively small compared to scenario (b). This smaller bias allows ES-CVTMLE to achieve moderate efficiency gains over regular TMLE, with the pooled parameter selected in about 30-70\% of its cross-validation folds on average. This enables ES-CVTMLE to leverage external data for efficiency improvements. However, in scenario (b), the bias magnitude introduced by the external data is significantly larger, nearly eliminating ES-CVTMLE's efficiency gain. Notice that ES-CVTMLE rejects external data in most cross-validation folds due to the large estimated bias. In contrast, A-TMLE effectively learns the simple bias function in this scenario, achieving efficiency gains while maintaining valid type I error control. Recall that in the Introduction section, we mentioned that an efficient estimator for the target parameter $\Psi$ may still offer no gain in the presence of external data, this is evident in the relative MSE plots with a regular TMLE that goes after $\Psi$ matching the MSE of a regular TMLE applied to the RCT data alone (red dashed line). Therefore, even theoretically optimal estimators may not yield efficiency gains from external data integration in finite sample.
\begin{figure}[h!]
\centering
\includegraphics[width=0.75\textwidth]{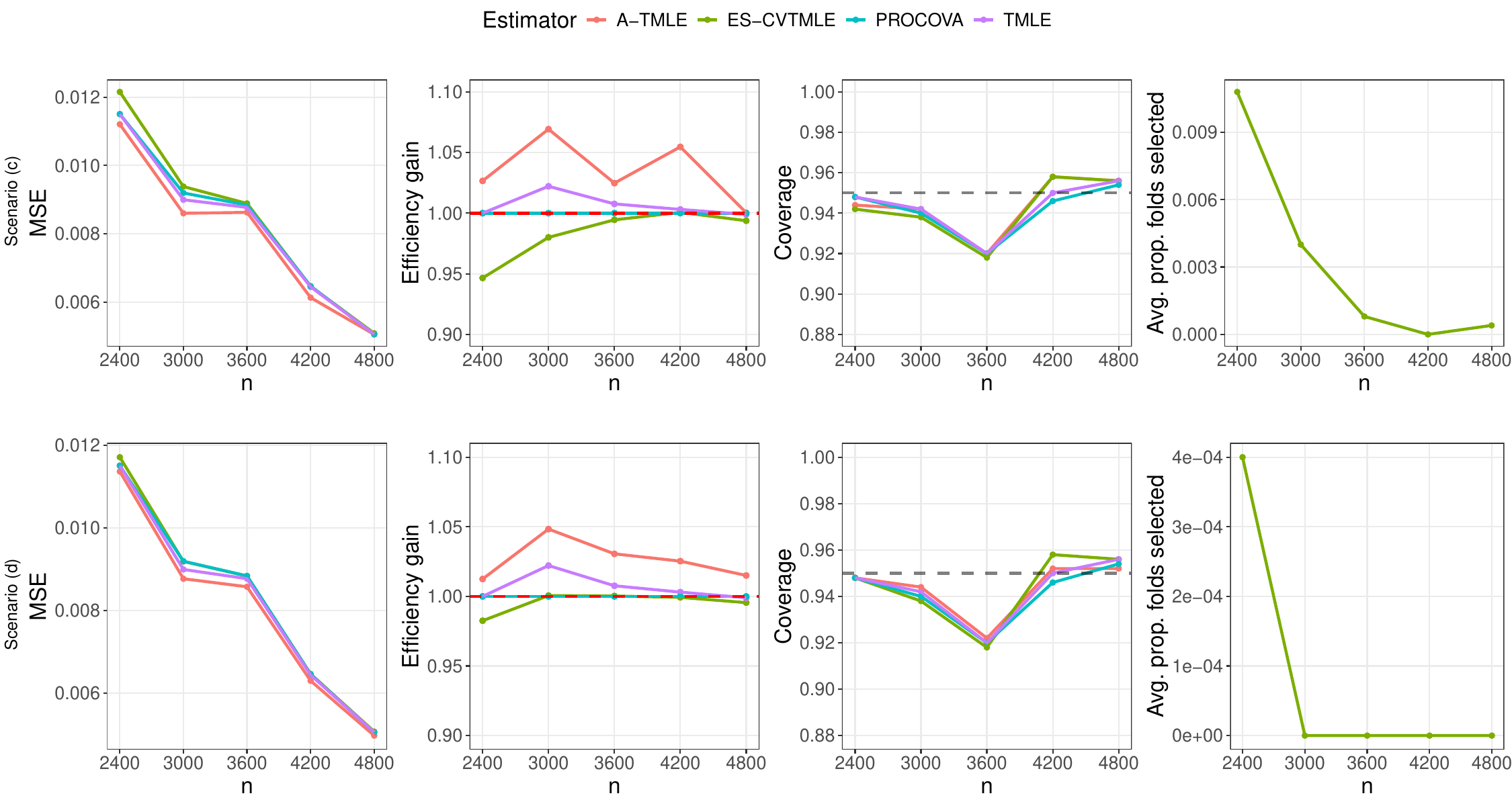}
\caption{Comparing the MSE, relative MSE, 95\% confidence interval coverage of A-TMLE, ES-CVTMLE, and TMLE estimator for increasing sample sizes of both the RCT and external data. The reference (red dashed line) in the relative MSE plots are the MSEs of the RCT-only estimator. The plots in the right-most column show the average proportion of cross-validation folds the ES-CVTMLE estimator pooled the external data. Scenarios (c) and (d) represent two different data-generating-processes where the conditional average RCT-enrollment effect involves indicator functions, higher-order polynomials, and higher-way interaction terms.}
\label{fig:complex}
\end{figure}

In some practical settings, the bias function may be highly non-linear and may not be well-approximated by a linear combination of main terms. Figure \ref{fig:complex} shows results for such scenarios. In particular, the true bias terms in scenarios (c) and (d) involve indicator functions, higher-order polynomials, and higher-way interactions. The data generating processes are provided in Appendix B. The use of more flexible nonparametric learning algorithms for estimating the bias function is necessary in these settings. We utilize HAL with a rich set of first-order spline basis functions to data-adaptively learn the bias working model. In scenarios with complex bias, ES-CVTMLE fails to achieve any efficiency gain, and at times, it underperforms by having a higher MSE compared to that of a regular TMLE on the RCT data alone. This may stem from not estimating the bias accurately or as a result of its cross-validation method not fully leveraging the data due to the requirements for bias and variance estimation for the selection criteria. On the other hand, A-TMLE continues to offer efficiency gain while having close to 95\% coverage. However, the efficiency gain is less compared to scenarios (a) and (b) with simpler bias functions, which is expected given the increased complexity of the bias. 

\subsection{Practical violations of trial enrollment positivity}
The causal interpretation of the target parameter $\Psi$ requires that the external population either coincides with or is a subset of the trial-eligible population \cite{dahabreh_extending_2019}. In Section \ref{sec:formulation}, we further emphasized that to achieve maximal efficiency in estimating the target parameter $\Psi$, the dependence of trial enrollment on patient baseline covariates should be made as weak as possible by design when selecting external patients. However, the positivity assumption (\ref{A3}), which requires that every patient has a non-zero probability of being enrolled in the trial, may still be violated practically. To evaluate the performance of A-TMLE under such conditions, we designed a simulation setup with trial enrollment near positivity violations, making the covariate distributions between the RCT and external data differ significantly. Specifically, we considered two scenarios where the smallest trial enrollment probability $P_0(S=1\mid W)$ was approximately 0.05 and 0.003, respectively.
\begin{figure}[h!]
\centering
\includegraphics[width=0.8\linewidth]{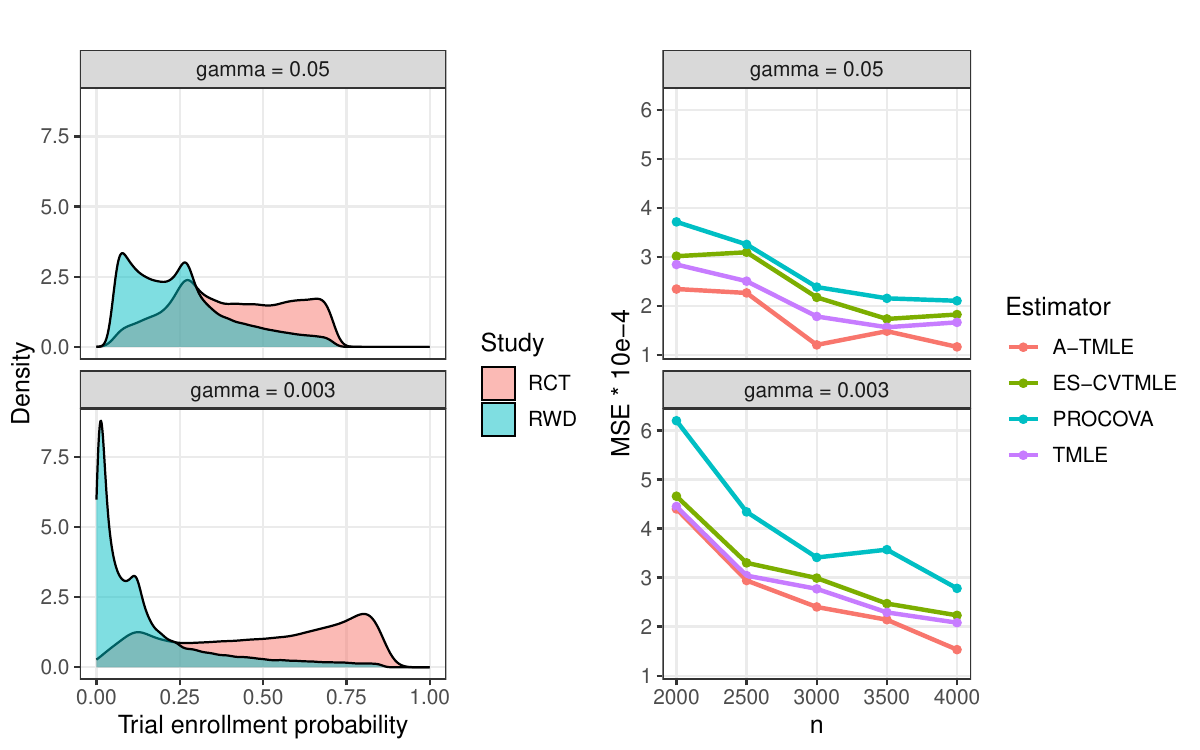}
\caption{Settings with trial enrollment positivity violations. The two subplots on the left display the densities of trial enrollment probabilities. The parameter $\gamma$ is defined as the infimum of the trial enrollment probability, $P_0(S=1\mid W)$.}
\label{fig:positivity}
\end{figure}
As shown in Figure \ref{fig:positivity}, A-TMLE remains to achieve the largest efficiency gain with the smallest MSE compared to other estimators. This is likely due to its avoidance of inverse weighting by the trial enrollment probability, as discussed in Section \ref{sec:psi_pound}, making it relatively more stable in near positivity violation settings. Interestingly, in this simulation setting, ES-CVTMLE never selects the pooled-ATE estimand but still achieves a moderate efficiency gain compared to an efficient estimator applied to the RCT data alone. This occurs because when ES-CVTMLE rejects external data, it estimates the parameter $\Psi$, which quantifies the treatment effect in the trial while averaging over the covariate distribution in the pooled data. In this scenario, because certain patient subgroups underrepresented in the trial are more prevalent in the external data, such increased representation allows ES-CVTMLE to achieve moderate efficiency gain by averaging over a more representative sample of the target population.

\section{Data application}\label{sec:data_app}
We applied our method to augment the DEVOTE trial with external data from Optum’s de-identified Clinformatics database. The DEVOTE trial is a double-blind, treat-to-target, event-driven study designed to evaluate the cardiovascular safety of degludec, an ultralong-acting, once-daily basal insulin approved by the FDA for various patient populations with diabetes \cite{marso_efficacy_2017}. Its primary objective was to determine whether degludec posed no greater cardiovascular risk than glargine, using a pre-specified noninferiority hazard ratio margin of 1.3 (allowing up to 30\% excess risk). The treatment arm ($n=3,818$) included patients receiving insulin degludec, while the comparator arm ($n=3,817$) consisted of patients receiving insulin glargine. The primary endpoint was a composite outcome for major adverse cardiac events (MACE), which included death from cardiovascular causes, nonfatal myocardial infarction, and nonfatal stroke. Our scientific question of interest focuses on the 365-, 540-, and 730-day (1-, 1.5-, 2-year) risk differences in MACE between the treatment and comparator arms.

Augmenting the trial with external data to improve efficiency in this context is especially compelling for several reasons. First, there is a hypothesis that insulin degludec may reduce the risk of MACE compared to insulin glargine. However, achieving sufficient statistical power to test for superiority often requires much larger cohorts and longer follow-up periods, typically three to five years, to observe relatively rare events such as MACE \cite{regier_more_2016}. Second, once drugs like insulin degludec meet the pre-specified noninferiority margin and receive FDA approval, substantial amounts of real-world data often become available. Such data may provide evidence of cardiovascular safety advantages of the given drug over its alternatives. In light of these findings, continuing to keep patients in the comparator arm without allowing treatment switches to a safer option might raise ethical concerns, particularly when the sole purpose for adhering to the assigned treatment status is to achieve sufficient power for a superiority test \cite{lin_propensity_2018}. These considerations highlight the potential benefits of integrating external data with RCTs. By augmenting the DEVOTE trial with real-world data, researchers may enhance statistical power, explore opportunities for establishing treatment superiority, and improve the generalizability of the trial findings \cite{dahabreh_extending_2019}.

Challenges in augmenting the DEVOTE trial with external data include inconsistencies in outcome measurement across data sources, non-concurrency, potential unmeasured confounding in the external data, and other differences between trial settings and routine care. For example, in the DEVOTE trial, adverse events were closely monitored and confirmed by an event-adjudication committee, whereas in the external data, some events may have been missed or inconsistently recorded. The trial concluded in 2017, but we extended the external data sampling timeframe to 2021 to ensure a decent number of patients receiving insulin degludec in the external cohort. To establish comparable cohorts between the RCT and observational data we emulate the DEVOTE trial following the framework detailed in \cite{hernan_using_2016}. We translated and applied the trials inclusion and exclusion criteria and used the initiation of insulin degludec or insulin glargine as the study start date. Although efforts have been made to address potential sources of bias, concerns about potential residual heterogeneity still remain, such as systematic differences in adherence between the controlled trial setting and routine clinical care. These differences can lead to non-comparable exposure definitions, violating the consistency assumption, which further underscores the necessity of applying our method that does not rely on assumptions beyond those that are known to likely hold true based on the design of the RCT. Censoring events were handled using the inverse probability of censoring weights technique documented in Appendix B.
\begin{figure}[h!]
\centering
\includegraphics[width=1\linewidth]{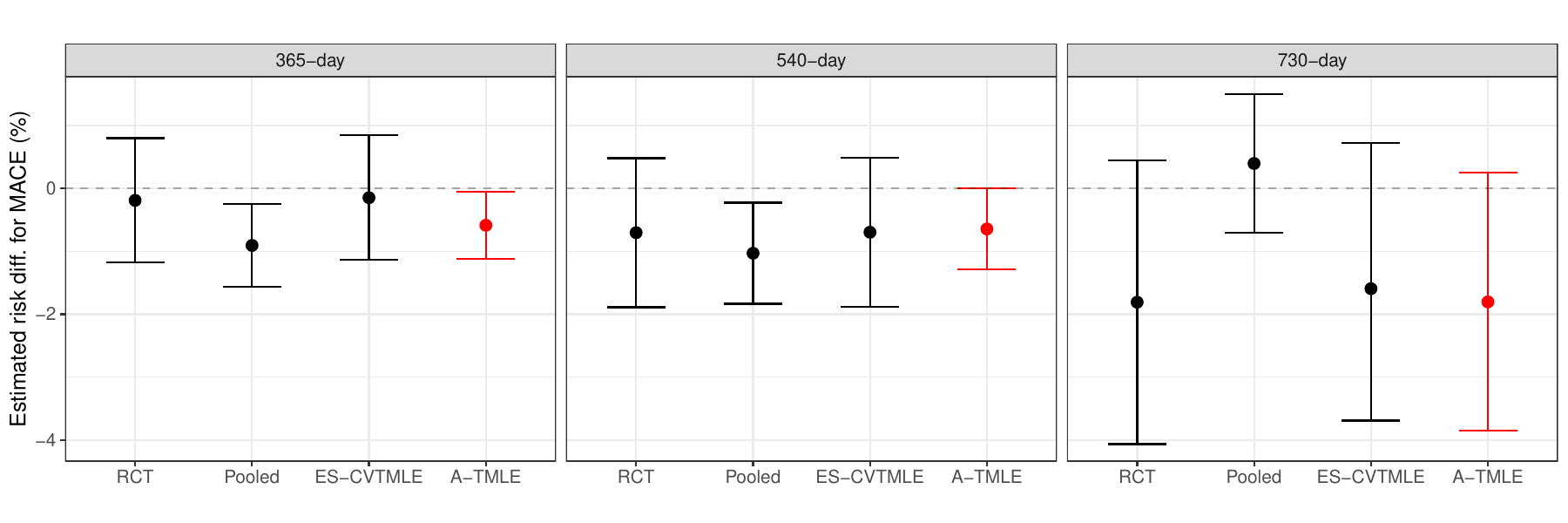}
\caption{Estimated 365-, 540-, and 730-day risk differences (\%) of MACE between the insulin degludec and insulin glargine groups, with corresponding 95\% confidence intervals. `RCT' refers to results from an augmented inverse propensity score weighted estimator applied to DEVOTE trial data only. `Pooled' refers to results from the same estimator applied to combined data from the DEVOTE trial and external Optum data. `ES-CVTMLE' and `A-TMLE' refer to results from the ES-CVTMLE and A-TMLE estimator applied to the pooled data, respectively.}
\label{fig:devote_optum}
\end{figure}
\begin{table}[h!]
\centering
\resizebox{\linewidth}{!}{
\begin{tabular}{cc|ccc}
\toprule
 & & \multicolumn{3}{c}{Estimated risk difference (\%) of MACE with 95\% CIs} \\
\midrule
Estimator & Study design & 365-day & 540-day & 730-day \\
\midrule
AIPW & DEVOTE & -0.190 (-1.177, 0.798) & -0.704 (-1.888, 0.481) & -1.809 (-4.065, 0.447) \\ 
AIPW & DEVOTE + Optum & -0.906 (-1.565, -0.248)$^\star$ & -1.031 (-1.834, -0.227)$^\star$ & 0.396 (-0.705, 1.497) \\
ES-CVTMLE & DEVOTE + Optum & -0.148 (-1.137, 0.842) & -0.697 (-1.883, 0.488) & -1.594 (-3.692, 0.718) \\ 
A-TMLE & DEVOTE + Optum & -0.586 (-1.118, -0.054)$^\star$ & -0.645 (-1.288, -0.001)$^\star$ & -1.804 (-3.855, 0.246) \\ 
\bottomrule
\end{tabular}
}
\caption{Estimated risk differences of MACE between the insulin degludec and insulin glargine groups, with corresponding 95\% confidence intervals. Note that the AIPW estimator applied to the pooled data may be biased. $^\star$ denotes statistical significance under $\alpha=0.05$ for a two-sided test.}
\label{tab:devote_optum}
\end{table}

Figure \ref{fig:devote_optum} shows the estimated 1-, 1.5-, and 2-year risk differences of MACE using different estimation strategies, with corresponding values presented in Table \ref{tab:devote_optum}. Specifically, the label `RCT' refers to estimates obtained using augmented inverse propensity weighted (AIPW) estimators applied to the RCT data. The label `Pooled' corresponds to AIPW estimates of pooled-ATE estimand $\tilde{\Psi}$, which, as discussed in Section \ref{sec:formulation}, may be biased had the mean exchangeability over studies assumption been violated. In fact, we do observe some evidence suggesting that this might be the case, since the `Pooled' effect estimate shows the opposite direction to the RCT estimate for 730-day. Even though the confidence interval overlaps between the two, large deviations from the RCT only estimate are still concerning because in this case, we expect RCT to give unbiased estimate, so the point estimates should be close. This further highlights the need to use estimators including ES-CVTMLE and A-TMLE which do not rely on those assumptions. As one could see, the ES-CVTMLE and A-TMLE point estimates are more aligned with the RCT estimate. Comparing the two, A-TMLE consistently produces narrower confidence intervals. For the 1- and 1.5-year risk estimates, A-TMLE provides statistically significant results supporting the superiority of insulin degludec in reducing MACE risk. The larger confidence intervals at the 2-year mark may be attributed to a higher number of censored patients, leading to increased variability in the estimates. We also provide the values for the point estimates and corresponding 95\% confidence intervals in Table \ref{tab:devote_optum}.

\section{Discussion}\label{sec:discussion}
In this article, we proposed an estimator for the average treatment effect using RCT data combined with RWD. We considered a target parameter that fully respects the RCT as the gold standard for treatment effect estimation without relying on extra identification assumptions beyond those that are known to be true for RCTs. The only additional assumption is that every individual in the pooled population must have a non-zero probability of participating in the RCT, a condition typically could be made more plausible during the sampling of external real-world patients. Despite focusing on a conservative target parameter for which an efficient estimator might still lack statistical power, our proposed A-TMLE demonstrated potential efficiency gains. These gains are primarily driven by the complexity of the bias working model rather than the bias magnitude, which is a desirable property, particularly in scenarios where the bias is expected to be large.

From a theoretical perspective, although A-TMLE gave up some regularity along perturbation paths of the data distribution that fall outside the learned working model, the oracle model, defined as the model that contains the true data distribution $P_0$, includes only relevant perturbations of the data distribution, so such loss of regularity is generally not a problem. Nonetheless, since the data-adaptive working model approximates the oracle model, going for a larger working model than what cross-validation suggests might result in finite sample bias reductions. Therefore, in practice, one may consider strategies including undersmoothing \cite{vdl_efficient_2023} or enforcing a minimal working model, such as a main-term only minimal working model, to further protect against finite sample bias. In this manner, the family of A-TMLEs is rich enough to determine a most sensible trade-off between regularity (i.e., possible bias) and efficiency gain.

We highlight some possible extensions of A-TMLE. First, since we do not impose extra assumptions, virtually any external data, even those with an outcome defined slightly differently from the RCT, could be utilized. This approach could extend to scenarios involving surrogate outcomes, where the difference in the outcomes might be adequately approximated as a function of covariates and treatment. Additionally, A-TMLE could be useful when researchers wish to augment a well-designed observational study, where the desired causal effect could be identified, with other observational data potentially subject to unmeasured confounding. Extending this approach to handle multiple external data sources could be achieved either by stratifying based on the choice of external study to learn a separate bias function for each external data source, or by pooling across the external studies to jointly learn a bias function. 

One limitation of A-TMLE is that, although theoretically the efficiency gain is driven by the complexity of the bias working model, it remains unclear in practice how much efficiency gain one should expect. For instance, our simulations have shown situations where A-TMLE is almost twice as efficient, yet other cases where the efficiency gain is small. Future research could focus on characterizing the expected efficiency gain in finite samples and identifying conditions under which substantial gains are most likely. Another challenge is determining how to enforce a minimal working model without compromising efficiency gains. Generally, adding more terms to the working model reduces bias but may limit potential efficiency gain. However, without such a minimal working model, one runs the risk of introducing bias if cross-validation overly favors a sparse model when the true model is actually complex. Addressing this may involve implementing undersmoothing methods that find the right trade-off between finite sample bias and efficiency. Even with robust estimation mechanisms in place, conducting a sensitivity analysis remains important to provide an additional layer of confidence. Future work could explore the possibility of integrating, for example, negative control outcomes (or treatments), as part of A-TMLE.

\section*{Acknowledgement}
The authors thank members of the Joint Initiative for Causal Inference (JICI) working group for insightful comments and discussions. We also extend our gratitude to the two anonymous reviewers for their valuable suggestions, which greatly improved the manuscript. This work was funded by Novo Nordisk.

\newpage
\printbibliography

\newpage
\newpage
\section*{Appendix A: Proof of Lemmas}
\begin{proof}[Proof of Lemma \ref{lem:gradient_psi_and_psi_2}]
First consider the case that we do not assume $S\perp W$ but possibly have a model on $g(A\mid S,W)$. Then the efficiency bound corresponds to the one in the nonparametric model due to the target parameter only depending on $P_W$ and $Q$, and that the tangent space of $g$ is orthogonal to the tangent space of these nuisance parameters. Therefore, it suffices to derive the influence curve of the empirical plug-in estimator acting as if $W$ is discrete, which then also establishes the general efficient influence curve by an approximation argument. The influence curve of the empirical plug-in estimator $\Psi(P_n)$ is straightforward to derive by using that the canonical gradient/influence curve of the empirical estimate of $E(Y\mid W=w,A=a,S=s)$ is given by  $E(Y\mid W=w,A=a,S=1)=I(S=1,A=a,W=w)/p(S=1,W=w,A=a)(Y-Q(S,W,A))$, while we use the empirical measure of $W_1,\ldots,W_n$ as estimator of $P_W$.  Consider now the statistical model that also assumes that $S\perp W$. The only way this model assumption can affect the estimator is that it might yield a different estimator of $P_W$ than the empirical measure. However, the empirical distribution of $W$ is still an efficient estimator of $P_W$, even when $S$ is independent of $W$. In other words, if we have two samples of $W$ from the same population, then the efficient estimator of the marginal distribution of $W$  is still the empirical measure of the combined sample on $W$.
The proof for $\Psi_2 $ is similar. If we assume $S\perp W$, then the canonical gradient for $\Psi_2$ changes a little due to an efficient estimator of  $P_{W\mid S=1}$ should now use the empirical measure of $W$ for the combined sample, instead of using the empirical measure of $W$, given $S=1$.
\end{proof}

\begin{proof}[Proof of Lemma \ref{lem:psi_tilde_loss}]
We have $\bar{Q}_P(W,A)=\tilde{\theta}_P(W)+A\tau_{A,P}(W)$, where $\tilde{\theta}_P(W)=E_P(Y\mid A=0,W)$ and $\tau_{A,P}(W)=E_P(Y\mid A=1,W)-E_P(Y\mid A=0,W)$. We have to compute a projection of $\bar{Q}_P$ onto the linear space $\tilde{\theta}(W)+A\sum_j\beta(j)\phi_j(W)$. We can write 
\[
\tilde{\theta}(W)+A\sum_j\beta(j)\phi_j(W)=\left(\tilde{\theta}(W)+g(1\mid W)\sum_j \beta(j)\phi_j(W)\right)+(A-g(1\mid W))\sum_j\beta(j)\phi_j(W).\]
Thus, the linear space on which we are projecting is an orthogonal sum of $L^2(P_W)$ and the linear span of $\{(A-g(1\mid W))\phi_j(W):j\}$.
In $L^2(P)$, functions of $W$ are orthogonal to the linear span of $\{(A-g(1\mid W))\phi_j(W):j\}$.
Therefore, the projection of $\bar{Q}_P$ on this orthogonal sum space  is given by $\theta_P(W)=E(\bar{Q}_P\mid W)=E_P(Y\mid W)$ plus the projection of $\bar{Q}_P=\tilde{\theta}_P+A \tau_{A,P}$ onto $\{(A-g(1\mid W))\phi_j(W):j\}$.
The projection of $\tilde{\theta}_P$ on this linear span equals zero since a function of $W$ is orthogonal to $(A-g(1\mid W))\phi_j(W)$. The projection of $A\tau_{A,P}$ is given by $(A-g(1\mid W))\sum_j \beta_{P}(j)\phi_j$ with  $\beta_P=\arg\min_{\beta}E_P(A\tau_{A,P}(W)-(A-g(1\mid W))\sum_j\beta(j)\phi_j(W))^2$.
We can write $A\tau_{A,P}=(A-g(1\mid W))\tau_{A,P}(W)+g(1\mid W)\tau_{A,P}(W)$. So 
$\beta_P=\arg\min_{\beta}E_P (A-g(1\mid W))^2(\tau_{A,P}(W)-\sum_j \beta(j)\phi_j(W))^2$, which equals
$\arg\min_{\beta}E_P g_P(1-g_P)(1\mid W)(\tau_{A,P}(W)-\sum_j\beta(j)\phi_j(W))^2$. This proves the lemma.
\end{proof}

\begin{proof}[Proof of Lemma \ref{lem:psi_tilde_beta_gradient}]
We want to derive the canonical gradient of $\beta_P$ defined on a nonparametric model.
Our starting point is that
\[
\beta_P=\arg\min_{\beta}E_P\left(Y-(A-g_P(1\mid W))\sum_j\beta(j)\phi_j(W)\right)^2.\]
This shows that $\beta(P)$ solves the equation
\[
U(P,g_P,\beta)=E_P (A-g_P(1\mid W)){\bf \phi}(W)\left(Y-(A-g_P(1\mid W))\sum_j \beta(j)\phi_j(W)\right)=0.\]
By the implicit function theorem,  for paths $\{P_{\epsilon}:\epsilon\}$ through $P$ at $\epsilon =0$, we have at $\epsilon =0$:
\[
\frac{d}{d\epsilon}\beta(P_{\epsilon})=-\frac{d}{d\beta}U(P,g,\beta)^{-1}\frac{d}{d\epsilon}U(\beta(P),P_{\epsilon},g_{\epsilon}).\]
Let $I_P=E_Pg(1-g)(W){\bf \phi}{\bf \phi}^{\top}$, then it follows that 
\[
D_{\beta,P}=I_P^{-1}\{D_{1,P}+D_{2,P}\},\]
where $D_{1,P}$ is the canonical gradient of $P\rightarrow U(\beta_1,P,g_1)$ at $\beta_1=\beta_P$ and $g_1=g_P$, and $D_{2,P}$ is the canonical gradient of $P\rightarrow U(\beta_1,P_1,g_P)$ at $\beta_1=\beta_P$ and $P_1=P$. 
Since $U(\beta_1,P,g_1)$ is just a mean parameter with respect to $P$ (like $E_PD$ which has canonical gradient $D-E_PD$), it follows that 
\[
D_{1,P}=(A-g_P(1\mid W)){\bf \phi}(W)\left(Y-(A-g_P(1\mid W))\sum_j \beta_{P}(j)\phi_j(W)\right).\]
Moreover, 
\[
U(\beta_1,P_1,g_P)=E_{P_1}( A-g_P(1\mid W)){\bf \phi}(W)\left(\bar{Q}_{P_1}-(A-g_P(1\mid W))\sum_j \beta_{P_1}(j)\phi_j\right).\]
Therefore, it remains to determine the canonical gradient of 
\begin{align*}
P&\rightarrow \int_{w,a} p_1(w,a) (a-g_P(1\mid w)){\bf \phi}(w)(\bar{Q}_1(w,a)-(a-g_P(1\mid w))\sum_j \beta_1(j)\phi_j(w))d\mu(w,a)\\
&=\int_{w,a}p_1(w,a) (a-g_P(1\mid w)){\bf \phi}(w)\bar{Q}_1(w,a)-
\int_{w,a} p_1(w,a)(a-g_P(1\mid w))^2{\bf \phi}(w)\sum_j \beta_1(j)\phi_j(w)
\end{align*}
at $p_1=p$.
We know that the canonical gradient of $P\rightarrow g_P(1\mid w)$ is given by $I(W=w)/p(w)(A-g_P(1\mid W))$. So, by the delta-method, the canonical gradient is given by
\begin{align*}
D_{2,P}(W,A)&=-\int_{w,a} p(w)g(a\mid w) I(W=w)/p(w)(A-g_P(1\mid W)){\bf \phi}(w) \bar{Q}_P(w,a)d\mu(w,a)\\
&+2\int_{w,a} p(w)g(a\mid w) (a-g_P(1\mid w)){\bf \phi}(w)I(W=w)/p(w)(A-g_P(1\mid W))\sum_j \beta_j\phi_j(w)d\mu(w,a)\\
&=-(A-g_P(1\mid W)){\bf \phi}(W)\int_a g(a\mid W)\bar{Q}_P(W,a)d\mu(a\mid w)\\
&+2(A-g_P(1\mid W)){\bf \phi}(W)\sum_j \beta_j\phi_j(W) \int_ag(a\mid W)(a-g_P(1\mid W))d\mu(a\mid w)\\
&=-(A-g_P(1\mid W)){\bf \phi}(W)\theta_P(W),
\end{align*}
since $\int_a g_P(a\mid W)(a-g_P(1\mid W))d\mu(a\mid w)=E_P(A-E_P(A\mid W))=0$.
Thus,
\begin{align*}
D_{1,P}+D_{2,P}&= (A-g_P(1\mid W)){\bf \phi}(W)\left(Y-(A-g_P(1\mid W))\sum_j \beta_P(j)\phi_j(W)\right)-(A-g_P(1\mid W)){\bf \phi}(W)\theta_P(W)\\
&=(A-g_P(1\mid W)){\bf \phi}(W)\left(Y-\theta_P(W)-(A-g_P(1\mid W))\sum_j \beta_P(j)\phi_j(W)\right).
\end{align*}
This proves that
\[
D_{\beta,P}^r=I_P^{-1} (A-g_P(1\mid W)){\bf \phi}(W)(Y-\theta_P(W)-(A-g_P(1\mid W))\sum_j \beta_P(j)\phi_j(W)),\]
where $I_P=E_Pg(1-g)(1\mid W){\bf \phi}{\bf \phi}^{\top}$.
\end{proof}

\begin{proof}[Proof of Lemma \ref{lem:psi_pound_loss}]
The proof is analogous to the proof of Lemma \ref{lem:psi_tilde_loss} by replacing $A$ with $S$. Specifically, in Lemma \ref{lem:psi_tilde_loss}, we showed the loss function for learning the working model for the conditional effect of $A$ on $Y$. Here, in Lemma \ref{lem:psi_pound_loss}, we have a similar loss function for learning the conditional effect of $S$ on $Y$.
\end{proof}

\begin{proof}[Proof of lemma \ref{lem:psi_pound_beta_gradient}]
The proof is analogous to the proof of Lemma \ref{lem:psi_tilde_beta_gradient} by replacing $A$ with $S$ and $W$ with $W,A$.
\end{proof}

\begin{proof}[Proof of lemma \ref{lem:psi_pound_gradient}]
The $p_{Y\mid S,W,A}$-score component of $D_{{\cal M}_{S,w},P}^{\#}$ is given by
\begin{align*}
D_{{\cal M}_{S,w},\beta,P}^{\#}&\equiv \left\{ E_P\Pi_P(0\mid W,0)\frac{d}{d\beta_P}\tau_{S,\beta_P}(W,0)\right\}^{\top} D_{\beta,P}\\
&-\left\{E_P \Pi_P(0\mid W,1)\frac{d}{d\beta_P}\tau_{S,\beta_P}(W,1)\right\}^{\top} D_{\beta,P}\\
&=\sum_j D_{\beta,P,j} E_P\Pi_P(0\mid W,0)\phi_j(W,0)-\sum_j D_{\beta,P,j} E_P \Pi_P(0\mid W,1)\phi_j(W,1).
\end{align*}
The $p_W$-score component of $D_{{\cal M}_{S,w},P}^{\#}$ is given by 
$$
D^{\#}_{{\cal M}_{S,w},P_W,P}\equiv \Pi_P(0\mid W,0)\tau_{S,\beta_P}(W,0)-\Pi_P(0\mid W,1) \tau_{S,\beta_P}(W,1) -\Psi_{{\cal M}_{S,w}}^{\#}(P).
$$
Finally, we need the contribution from the dependence of $\Psi^{\#}_{{\cal M}_{S,w}}(P)$ on the conditional distribution $\Pi_P$.
The influence curve of $\Pi_P(0\mid w,a)$ is given by $D_{{\cal M}_{S,w},\Pi,(w,a),P}=I(W=w,A=a)/P(w,a)(I(S=0)-\Pi_P(0\mid w,a))$.
So, $\Pi_P$-score component of $D_{{\cal M}_{S,w},P}^{\#}$ is given by
\begin{align*}
D^{\#}_{{\cal M}_{S,w},\Pi,P}&\equiv \int_w I(W=w,A=0)/P(w,0)\tau_{S,\beta_P}(w,0)  (I(S=0)-\Pi_P(0\mid W,0))dP(w)\\
&-\int_{w} I(W=w,A=1)/P(w,1)\tau_{S,\beta_P}(w,1) (I(S=0)-\Pi_P(0\mid W,1))dP(w) \\
&=\left\{  \frac{I(A=0)}{g_P(0\mid W)} \tau_{S,\beta_P}(W,0)-\frac{I(A=1)}{g_P(1\mid W)} \tau_{S,\beta_P}(W,1)\right\}(I(S=0)-\Pi_P(0\mid W,A)).
\end{align*}
Now, note that $I(S=0)-\Pi(0\mid W,A)=-(I(S=1)-\Pi(1\mid W,A))$. Therefore, we have
$$
D^{\#}_{{\cal M}_{S,w},\Pi,P}=\left\{  \frac{A}{g_P(1\mid W)} \tau_{S,\beta_P}(W,0)-\frac{1-A}{g_P(0\mid W)} \tau_{S,\beta_P}(W,1)\right\}(S-\Pi_P(1\mid W,A)).
$$
So, we have every component in
\[
D^{\#}_{{\cal M}_{S,w},P}=D^\#_{{\cal M}_{S,w},P_W,P}+D^\#_{{\cal M}_{S,w},\Pi,P}+D^\#_{{\cal M}_{S,w},\beta,P}.\]
\end{proof}

\section*{Appendix B: Data-Generating Processes, Missing Outcome Data Structure}
The data-generating process for scenarios (a) and (b) in Section \ref{sec:simulation} is:
\begin{align*}
U_Y&\sim\mathcal{N}(0,1)\\
W_1&\sim\mathcal{N}(0,1)\\
W_2&\sim\mathcal{N}(0,1)\\
W_3&\sim\mathcal{N}(0,1)\\
A&\sim \left\{
\begin{array}{ll}
\text{Bern}(0.67) & \text{if } S=1 \\
\text{Bern}(\text{expit}(0.5W_1)) & \text{if } S=0
\end{array}
\right.\\
B&= \left\{
\begin{array}{ll}
0.2+1.1W_1\cdot I(A=0) & \text{scenario (a)} \\
0.5+3.1W_1\cdot I(A=0)+0.8W_3 & \text{scenario (b)}
\end{array}
\right.\\
Y &= 2.5+0.9W_1+1.1W_2+2.7W_3+1.5A+U_Y+I(S=0)\cdot B.
\end{align*}
The data-generating process for scenarios (c) and (d) in Section \ref{sec:simulation} is:
\begin{align*}
U_Y&\sim\mathcal{N}(0,1)\\
W_1&\sim\text{Uniform}(0,1)\\
W_2&\sim\text{Uniform}(0,1)\\
W_3&\sim\text{Uniform}(0,1)\\
A&\sim \left\{
\begin{array}{ll}
\text{Bern}(0.67) & \text{if } S=1 \\
\text{Bern}(\text{expit}(W_1)) & \text{if } S=0
\end{array}
\right.\\
B&= \left\{
\begin{array}{ll}
0.3+0.9W_2\cdot I(A=0)+0.7W_3\cdot I(W_2 > 0.5) & \text{scenario (c)} \\
0.3+1.1W_1\cdot I(A=0)+0.9W_2^2W_3 & \text{scenario (d)}
\end{array}
\right.\\
Y &= 1.9+4.2A+0.9W_1+1.4W_2+2.1W_3+U_Y+I(S=0)\cdot B.
\end{align*}
The data-generating process for scenarios with practical violations of trial enrollment positivity in Section \ref{sec:simulation} is:
\begin{align*}
U_Y&\sim\mathcal{N}(0,0.2^2)\\
W_1&\sim\text{Uniform}(-1,1)\\
W_2&\sim\text{Uniform}(-1,1)\\
W_3&\sim\text{Uniform}(-1,1)\\
S&\sim\text{Bern}(\text{expit}(\alpha(-2+W_1+W_2+\sin(2W_1)+\sin(2W_2)))), \ \alpha\in\{0.5,1\}\\
A&\sim \left\{
\begin{array}{ll}
\text{Bern}(0.67) & \text{if } S=1 \\
\text{Bern}(\text{expit}(-0.5W_1)) & \text{if } S=0
\end{array}
\right.\\
B&=0.2+2.1W_1A\\
Y&=1.9+1.5A+0.9W_1+1.4W_2+2.1W_3+U_Y+I(S=0)\cdot B.
\end{align*}
When the outcome is subject to missingness, we consider the more general missing data structure $O=(S,W,A,\Delta,\Delta\cdot Y)$ where $\Delta$ is an indicator that equals to one if the outcome is observed. In this case, we consider the target parameter given by
$$
\Psi(P)=E_WE(Y_1-Y_0\mid S=1,W,\Delta=1).
$$
In other words, we are interested in the causal effect in the world where everyone's outcome is available. The identification of this target parameter requires an additional coarsening at random (or missing at random) assumption, which states that $(Y_1,Y_0)\perp \Delta\mid W,S=1$. Together with assumptions \ref{A1}, \ref{A2} and \ref{A3} from Section \ref{sec:formulation}, we have the target estimand given by
$$
\Psi(P_0)=E_0[E_0(Y\mid S=1,W,A=1,\Delta=1)-E_0(Y\mid S=1,W,A=0,\Delta=1)].
$$
Now, to learn the working models for the effect of $S$ on $Y$ and the effect of $A$ on $Y$, we use the following loss functions, respectively:
\begin{align*}
L_{\bar{Q}_P,\Pi_P,g^\Delta_P}(\tau)&=\frac{\Delta}{g^\Delta_P(1\mid S,W,A)}\left(Y-\bar{Q}_P(W,A)-(S-\Pi_P(1\mid W,A))\tau\right)^2,\\
L_{\theta_P,g_P,\tilde{g}^\Delta_P}(\tau)&=\frac{\Delta}{\tilde{g}^\Delta_P(1\mid W,A)}\left(Y-\theta_P(W)-(A-g_P(1\mid W))\tau\right)^2,
\end{align*}
where $g^\Delta_P(1\mid S,W,A)\equiv P(\Delta=1\mid S,W,A)$ and $\tilde{g}^\Delta_P(1\mid W,A)\equiv P(\Delta=1\mid W,A)$. In addition, we also multiply the factor $\Delta/g^\Delta_P(1\mid S,W,A)$ and $\Delta/\tilde{g}^\Delta_P(1\mid W,A)$ to the $Y$-components of the canonical gradients of the projection parameters we showed in Lemma \ref{lem:psi_tilde_gradient} and \ref{lem:psi_pound_gradient} respectively.

\section*{Appendix C: Adaptive CV-TMLE}
The A-TMLE is targeting a data dependent target parameter $\Psi_{{\cal M}_{w}}(P_0)$ ignoring that the parameter itself already depends on the data. In the literature for data-dependent target parameters, we have generally recommended the use of cross-validated TMLE (CV-TMLE) to minimize reliance on Donsker class conditions. Therefore, in this appendix, we describe and analyze the adaptive CV-TMLE for this particular type of data-adaptive target parameter $P_n\rightarrow \Psi_{{\cal M}_w(P_n)}(P_0)$. 

Let $P_{n,v}$ and $P_{n,v}^1$ be the empirical measures of the training and validation sample, respectively, for $v=1,\ldots,V$, using $V$-fold cross-validation scheme. Let ${\cal M}_{w,n,v}={\cal M}_w(P_{n,v})$. This defines a data-adaptive target parameter $\Psi_{{\cal M}_{w,n,v}}:{\cal M}\rightarrow\openr$.
Let $P_{n,v}^0$ be an initial estimator of $P_0$ based on $P_{n,v}$. Let $P_{n,v}^\star$ be a TMLE of $\Psi_{{\cal M}_{w,n,v}}(P_0)$ based on this initial estimator $P_{n,v}^0$, where the TMLE-update is applied to the validation sample $P_{n,v}^1$, so that $P_{n,v}^1 D_{{\cal M}_{w,n,v},P_{n,v}^\star}=o_P(n^{-1/2})$, for $v=1,\ldots,V$. One could also carry out a pooled-TMLE-update step by minimizing the cross-validated empirical risk $\epsilon\rightarrow 1/V\sum_v P_{n,v}^1 L(P_{n,v}^0(\epsilon))$ using the same least favorable path $\{P_{n,v}^0(\epsilon):\epsilon\}$. In both cases, we end up with $P_{n,v}^\star$, $v=1,\ldots,V$, that solves $1/V\sum_{v=1}^V P_{n,v}^1 D_{{\cal M}_{n,w,v},P_{n,v}^\star}=o_P(n^{-1/2})$. We note that if ${\cal M}_{w,n,v}$ is a parametric model, and a sample size of $n/V$ is large enough to train the parameters of this working model, then one could estimate $P_0$ with an MLE $P_{n,v}^\star=\arg\min_{P\in {\cal M}_{w,n,v}}P_{n,v}^1 L(P)$, thus not using an initial estimator. Using a least favorable path with many more extra parameters in the pooled-TMLE allows one to use a simple initial estimator (less data-adaptive), the pooled-TMLE starts resembling a parametric MLE over the cross-validated empirical risk. 

The adaptive CV-TMLE of $\Psi(P_0)$ is defined as
$$
\psi_n^\star\equiv \frac{1}{V}\sum_{v=1}^V \Psi_{{\cal M}_{w,n,v}}(P_{n,v}^\star).
$$
We note that this is the same as the CV-TMLE of the data adaptive target parameter $1/V\sum_{v=1}^V \Psi_{{\cal M}_{w,n,v}}(P_0)$ as proposed and analyzed in the targeted learning literature for general data-adaptive target parameters but applied to this particular type of data-adaptive target parameter \cite{hubbard_statistical_2016}.

We now analyze the adaptive CV-TMLE of $\Psi(P_0)$, analogue to the analysis of the A-TMLE in Section \ref{sec:theory}. 
As usual, for CV-TMLE we have
\[
\psi_n^\star-\frac{1}{V}\sum_{v=1}^V \Psi_{{\cal M}_{w,n,v}}(P_0)=
\frac{1}{V}\sum_{v=1}^V(P_{n,v}^1-P_0)D_{{\cal M}_{w,n,v},P_{n,v}^\star}+\frac{1}{V}\sum_{v=1}^V R_{{\cal M}_{w,n,v}}(P_{n,v}^\star,P_0)+o_P(n^{-1/2}).\]
We assume the analogue of $\Psi_{{\cal M}_{w}}(P_0)-\Psi(P_0)=o_P(n^{-1/2})$ as presented in in Section \ref{sec:theory}, which is now given by
\[
\frac{1}{V}\sum_{v=1}^V\{ \Psi_{{\cal M}_{w,n,v}}(P_0)-\Psi(P_0)\}=o_P(n^{-1/2}).\]
The sufficient condition for this is given by the following. 
Let $P_{0,n,v}\equiv \Pi_{{\cal M}_{w,n,v}}(P_0)$; $\tilde{P}_{0,n,v}=\Pi(P_{0,n,v}\mid {\cal M}_0)$. 
Let  $\tilde{D}_{{\cal M}_{w,n,v},P_{0,n,v}}$ be the projection of $D_{{\cal M}_0,P_{0,n,v}}$ onto the tangent space $T_{{\cal M}_{w,n,v}}(P_{0,n,v})\subset L^2_0(P_{0,n,v})$ of the working model ${\cal M}_{w,n,v}$ at $P_{0,n,v}$. 
Assume that ${\cal M}_{w,n,v}$ approximates ${\cal M}_0$ in the sense that $1/V\sum_{v=1}^V\{ \Psi(\tilde{P}_{0,n,v})-\Psi(P_{0,n,v})\}=o_P(n^{-1/2})$ and 
\[
\frac{1}{V}\sum_{v=1}^V (\tilde{P}_{0,n,v}-P_0)\left\{ D_{{\cal M}_0,\tilde{P}_{0,n,v}}-\tilde{D}_{{\cal M}_{w,n,v},\tilde{P}_{0,n,v}}\right\}=o_P(n^{-1/2}).
\]
Given this we have
\[
\psi_n^\star-\Psi(P_0)=
\frac{1}{V}\sum_{v=1}^V(P_{n,v}^1-P_0)D_{{\cal M}_{w,n,v},P_{n,v}^\star}+\frac{1}{V}\sum_{v=1}^V R_{{\cal M}_{w,n,v}}(P_{n,v}^\star,P_0)+o_P(n^{-1/2}).\]
The analogue rate of convergence condition is given by $1/V\sum_{v=1}^V R_{{\cal M}_{w,n,v}}(P_{n,v}^\star,P_0)=o_P(n^{-1/2})$. 
If conditional on the training sample $P_{n,v}$, $D_{{\cal M}_{w,n,v},P_{n,v}^\star}$ falls in a $P_0$-Donsker class (note that only the targeting step depends on $P_{n,v}^1$, making this a very weak condition), then it follows
\[
\psi_n^\star-\Psi(P_0)=O_P(n^{-1/2}).\]
Suppose $\pl D_{{\cal M}_{w,n,v},P_{n,v}^\star}-D_{{\cal M}_{w,n,v},P_{n,v,0}^\star}\pl_{P_0}=o_P(1)$, where $P_{n,v,0}^\star$ is the TMLE-update of $P_{n,v}^0$ under  $P_0$ (i.e., maximize $P_0 L(P_{n,v,\epsilon}^0)$ instead of $P_{n,v}^1 L(P_{n,v,\epsilon}^0)$ for a universal least favorable path through $P_{n,v}^0$). 
Then, we obtain
\[
\psi_n^\star-\Psi(P_0)=
\frac{1}{V}\sum_{v=1}^V(P_{n,v}^1-P_0)D_{{\cal M}_{w,n,v},P_{n,v,0}^\star}+o_P(n^{-1/2}).\]
For each $v$, we have that $(P_{n,v}^1-P_0)D_{{\cal M}_{w,n,v},P_{n,v,0}^\star}$ is a sum of mean zero independent random variables so that this term will be asymptotically normal if the variance converges to a fixed $\sigma^2_0$. Let's assume that $\pl D_{{\cal M}_{w,n,v},P_{n,v,0}^\star}-D_{{\cal M}_{w,n,v},P_0}\pl_{P_0}=o_P(1)$ so that we obtain
\[
\psi_n^\star-\Psi(P_0)=
\frac{1}{V}\sum_{v=1}^V(P_{n,v}^1-P_0)D_{{\cal M}_{w,n,v},P_0}+o_P(n^{-1/2}).\]
Let $\sigma_{w,n,v}^2=P_0 D_{{\cal M}_{w,n,v},P_0}^2$. Then we can write the leading term as:
\[
\frac{1}{V}\sum_{v=1}^V \sigma_{w,n,v} (P_{n,v}^1-P_0) D_{{\cal M}_{w,n,v},P_0}/\sigma_{w,n,v},\]
Let $\sigma_{w,n}\equiv 1/V\sum_{v=1}^V \sigma_{w,n,v}$. Thus,
\[
\sigma_{w,n}^{-1}(\psi_n^\star-\psi_0)=\sigma_{w,n}^{-1} \frac{1}{V}\sum_{v=1}^V (P_{n,v}^1-P_0) D_{{\cal M}_{w,n,v},P_0}/\sigma_{w,n,v}+o_P(n^{-1/2}).\]
Note that the right-hand side leading term is a weighted average over $v$ of sample means over $P_{n,v}^1$ of mean zero and variance one independent random variables, where each $v$-specific sample mean converges to $N(0,1)$. Therefore, it appears a rather weak condition to assume
\[\sigma_{w,n}^{-1} \frac{1}{V}\sum_{v=1}^V (P_{n,v}^1-P_0) D_{{\cal M}_{w,n,v},P_0}/\sigma_{w,n,v}\Rightarrow_d N(0,1).\]
In this manner, we have minimized the condition on ${\cal M}_w(P_n)$ with respect to convergence to the fixed oracle model, while still obtaining asymptotic normality. If we make the stronger assumption that $\pl D_{{\cal M}_{w,n,v},P_0}-D_{{\cal M}_0,P_0}\pl_{P_0}=o_P(1)$, then we obtain asymptotic linearity with a super-efficient influence curve: 
\[
\psi_n^\star-\Psi(P_0)=P_n D_{{\cal M}_0,P_0}+o_P(n^{-1/2}).\]
This proves the following theorem for the adaptive CV-TMLE of $\Psi(P_0)$.
\begin{theorem}\label{atmletheorem2}
Assume $1/V\sum_v P_{n,v}^1 D_{{\cal M}_{w,n,v},P_{n,v}^\star}=o_P(n^{-1/2})$. 
We have
\[
\psi_n^\star-\frac{1}{V}\sum_{v=1}^V \Psi_{{\cal M}_{w,n,v}}(P_0)=
\frac{1}{V}\sum_{v=1}^V(P_{n,v}^1-P_0)D_{{\cal M}_{w,n,v},P_{n,v}^\star}+\frac{1}{V}\sum_{v=1}^V R_{{\cal M}_{w,n,v}}(P_{n,v}^\star,P_0)+o_P(n^{-1/2}).\]
{\bf Oracle model approximation condition:}
Assume \[
\frac{1}{V}\sum_{v=1}^V\{ \Psi_{{\cal M}_{w,n,v}}(P_0)-\Psi(P_0)\}=o_P(n^{-1/2}).\]
A  sufficient condition for this is given by the following. 
Let $P_{0,n,v}\equiv \Pi_{{\cal M}_{w,n,v}}(P_0)$; $\tilde{P}_{0,n,v}=\Pi(P_{0,n,v}\mid {\cal M}_0)$. 
Let  $\tilde{D}_{{\cal M}_{w,n,v},P_{0,n,v}}$ be the projection of $D_{{\cal M}_0,P_{0,n,v}}$ onto the tangent space $T_{{\cal M}_{w,n,v}}(P_{0,n,v})\subset L^2_0(P_{0,n,v})$ of the working model ${\cal M}_{w,n,v}$ at $P_{0,n,v}$. 
Assume that ${\cal M}_{w,n,v}$ approximates ${\cal M}_0$ in the sense that $1/V\sum_{v=1}^V\{ \Psi(\tilde{P}_{0,n,v})-\Psi(P_{0,n,v})\}=o_P(n^{-1/2})$ and 
\[
\frac{1}{V}\sum_{v=1}^V (\tilde{P}_{0,n,v}-P_0)\left\{ D_{{\cal M}_0,\tilde{P}_{0,n,v}}-\tilde{D}_{{\cal M}_{w,n,v},\tilde{P}_{0,n,v}}\right\}=o_P(n^{-1/2}).
\]
Then,
\[
\psi_n^\star-\Psi(P_0)=
\frac{1}{V}\sum_{v=1}^V(P_{n,v}^1-P_0)D_{{\cal M}_{w,n,v},P_{n,v}^\star}+\frac{1}{V}\sum_{v=1}^V R_{{\cal M}_{w,n,v}}(P_{n,v}^\star,P_0)+o_P(n^{-1/2}).\]
{\bf Rate of convergence condition:}
Assume $1/V\sum_{v=1}^V R_{{\cal M}_{w,n,v}}(P_{n,v}^\star,P_0)=o_P(n^{-1/2})$. \newline
{\bf Weak Donsker class condition:}
Assume, conditional on the training sample $P_{n,v}$, $D_{{\cal M}_{w,n,v},P_{n,v}^\star}$ falls in a $P_0$-Donsker class (note that only the targeting step depends on $P_{n,v}^1$, making this a very weak condition). Then,
\[
\psi_n^\star-\Psi(P_0)=O_P(n^{-1/2}).\]
\newline
{\bf Consistency of TMLE $P_{n,v}^\star$ to $P_0$:}
Assume that $\pl D_{{\cal M}_{w,n,v},P_{n,v,0}^\star}-D_{{\cal M}_{w,n,v},P_0}\pl_{P_0}=o_P(1)$.
Let $\sigma_{w,n,v}^2=P_0 D_{{\cal M}_{w,n,v},P_0}^2$ and $\sigma_{w,n}=1/V\sum_{v=1}^V \sigma_{w,n,v}$.  
Then,
\[
\sigma_{w,n}^{-1}(\psi_n^\star-\psi_0)=\sigma_{w,n}^{-1} \frac{1}{V}\sum_{v=1}^V (P_{n,v}^1-P_0) D_{{\cal M}_{w,n,v},P_0}/\sigma_{w,n,v}+o_P(n^{-1/2}).\]
Note that the right-hand side leading term is a weighted average over $v$ of sample means over $P_{n,v}^1$ of mean zero and variance one independent random variables, so that each $v$-specific sample mean converges in distribution to $N(0,1)$.\newline
{\bf Asymptotic normality condition:} Assume
\[\sigma_{w,n}^{-1} \frac{1}{V}\sum_{v=1}^V (P_{n,v}^1-P_0) D_{{\cal M}_{w,n,v},P_0}/\sigma_{w,n,v}\Rightarrow_d N(0,1).\]
Then, \[
\sigma_{w,n}^{-1}(\psi_n^\star-\psi_0)\Rightarrow_d N(0,1).\]
{\bf Asymptotic linearity condition:}
Assume $\pl D_{{\cal M}_{w,n,v},P_0}-D_{{\cal M}_0,P_0}\pl_{P_0}=o_P(1)$. Then we have asymptotic linearity with a super-efficient influence curve: 
\[
\psi_n^\star-\Psi(P_0)=P_n D_{{\cal M}_0,P_0}+o_P(n^{-1/2}).\]
\end{theorem}

\section*{Appendix D: Alternative Methods for Constructing Data-Adaptive Working Models}
In Section \ref{sec:implementation}, we described how one could use relaxed-HAL to data-adaptively learn a working model for both the conditional effect of $S$ and $A$. In this appendix, we discuss alternative approaches for constructing working models.

We could use any adaptive estimator of $\theta_0=E_0(Y\mid S=0,W,A)$, which could be an HAL-MLE or super-learner stratifying by $S=0$. Let $\theta_n=\hat{\Theta}(P_n)$ be this estimator. We could then run an additional HAL on basis functions $\{S\phi_j(W,A): j\}$ to fit $\bar{Q}_0=\theta_n+\sum_j \beta_jS\phi_j$ using $\theta_n$ as off-set, and select the $L_1$-norm of $\beta$ of this last HAL-MLE with cross-validation. The resulting fit $\sum_j\beta_{j,n}S\phi_j$ now defines our data adaptive semiparametric regression working model ${\cal Q}_{w,sp,n}=\{\theta+\sum_{j,\beta_n(j)\not =0}\beta_j S\phi_j: \beta,\theta\}$ for $\bar{Q}_0$ that leaves $\theta$ unspecified and uses $\sum_j \beta_j S\phi_j$ for modeling $\tau_0$.
 
{\bf  Meta-HAL-MLE to generate working model for $\tau_0$:}
To obtain extra signal in the data for the second regression targeting $\tau_0$ we might apply cross-fitting again, analogue to above for the parametric regression working model.
 Let $P_{n,v}$ and $P_{n,v}^1$ be the training and validation sample for the $v$-th sample split, respectively, and $\theta_{n,v}=\hat{\Theta}(P_{n,v})$ be the resulting training sample fit of $\theta_0$, $v=1,\ldots,V$. We can define $\beta_n^C=\arg\min_{\beta,\pl \beta\pl\leq C} 1/V\sum_{v=1}^V P_{n,v}^1 L(\theta_{n,v}+\sum_j\beta_jS\phi_j)$ as the minimizer of the cross-fitted empirical risk. We can select $C$ with cross-validation resulting in a $\beta_n^{cf}$ which then generates the working model $\{\sum_{j,\beta_n^{cf}(j)\not =0}S\beta_j\phi_j:\beta\}$ for $\tau_0$. One can think of $\theta_n+\sum_j\beta_n^{cf}(j)S\phi_j$ as a meta-HAL-MLE, where cross-validation selects the best $\beta$-specific estimator $\hat{\theta}(P_n)+\sum_j S\beta_j\phi_j$, under an $L_1$-norm constraint on $\beta$.  
 Carrying out this meta-HAL-MLE (using cross-validation to select $C$) implies the semiparametric working model ${\cal Q}_{w,sp,n}=\{\theta+\sum_{j,\beta_n^{cf}(j)\not =0}\beta_jS\phi_j: \beta,\theta\}$. 
  
 {\bf Discrete super-learner to select among subspace specific meta-HAL-MLEs of $\tau_0$:}
  Analogue to this method for the parametric working model, we can use a discrete super-learner with a collection of the above described subspace-specific meta-HAL-MLEs. That will select one specific subspace-specific meta-HAL-MLEs with a corresponding working model for $\tau_0$. This then implies the working semiparametric regression model. Instead of using internal cross-validation within the subspace-specific meta-HAL-MLE for selecting $C$, one could also just use a discrete super-learner with candidate meta-HAL-MLEs that are both indexed by the subspace and the $C$ across the desired collection of subspaces and $C$-values. 
  
Once the working model ${\cal Q}_{w,n}$ is selected, we can apply our TMLEs for that choice of parametric or semiparametric regression working model. In that TMLE one could use the same initial estimator $\theta_n$ as was used to generate the working model. 
 
\subsection*{Learning the working model on an independent data set}
Imagine that one has access to a previous related study that collected the same covariates and outcome with possibly a different treatment. One could use this previous study to learn a working model for the outcome regression $\bar{Q}_0$, possibly apply some undersmoothing to make it not too adaptive towards the true regression function in that previous study. One wants to make sure that that study has a sample size larger or equal than the one in the current study so that the resulting working model is flexible enough for the sample size of the current study. One could now use this ${\cal M}_{w}$ as working model in the A-TMLE. One can now either use a TMLE or CV-TMLE of the target parameter $\Psi_{{\cal M}_{w}}(P_0)$. Below we present the theorem for the resulting adaptive CV-TMLE. The advantage of learning the working model on an external data set is that it allows us to establish asymptotic normality without requiring that $D_{{\cal M}_{w},P_0}$ converges to a fixed $D_{{\cal M}_0,P_0}$. One still needs that ${\cal M}_{w}$ approximates $P_0$ in the sense that $\Psi_{{\cal M}_{w}}(P_0)-\Psi(P_0)=o_P(n^{-1/2})$, but as we argued in \cite{adml} that condition can be achieved without relying on ${\cal M}_{w}$ to converge to a fixed oracle model ${\cal M}_0$ (in essence only relying on ${\cal M}_{w}$ to approximately capture $P_0$).

\begin{theorem}\label{atmletheorem3}
Let ${\cal M}_{w}$ be a fixed working model (learned on external data set). Let $\psi_n^\star=1/V\sum_{v=1}^V \Psi_{{\cal M}_{w}}(P_{n,v}^\star)$ be the CV-TMLE of $1/V\sum_{v=1}^V \Psi_{{\cal M}_{w}}(P_0)$ satisfying $1/V\sum_v P_{n,v}^1 D_{{\cal M}_{w},P_{n,v}^\star}=o_P(n^{-1/2})$.
We have
\[
\psi_n^\star-\Psi_{{\cal M}_{w}}(P_0)=
\frac{1}{V}\sum_{v=1}^V(P_{n,v}^1-P_0)D_{{\cal M}_{w},P_{n,v}^\star}+\frac{1}{V}\sum_{v=1}^V R_{{\cal M}_{w}}(P_{n,v}^\star,P_0)+o_P(n^{-1/2}).\]
{\bf Oracle model approximation condition:}
Assume \[
 \Psi_{{\cal M}_{w}}(P_0)-\Psi(P_0)=o_P(n^{-1/2}).\]
A  sufficient condition for this is presented in Theorem \ref{atmletheorem1}:
Let $P_{0,n}\equiv \Pi_{{\cal M}_{w}}(P_0)$; $\tilde{P}_{0,n}=\Pi(P_{0,n}\mid {\cal M}_0)$ for a submodel ${\cal M}_0\subset{\cal M}$ containing $P_0$. 
Let  $\tilde{D}_{{\cal M}_{w},P_{0,n}}$ be the projection of $D_{{\cal M}_0,P_{0,n}}$ onto the tangent space $T_{{\cal M}_{w}}(P_{0,n})\subset L^2_0(P_{0,n})$ of the working model ${\cal M}_{w}$ at $P_{0,n}$. 
Assume that ${\cal M}_{w}$ approximates $P_0$ in the sense that $\Psi(\tilde{P}_{0,n})-\Psi(P_{0,n})=o_P(n^{-1/2})$ and 
\[
(\tilde{P}_{0,n}-P_0)\left\{ D_{{\cal M}_0,\tilde{P}_{0,n}}-\tilde{D}_{{\cal M}_{w},\tilde{P}_{0,n}}\right\}=o_P(n^{-1/2}).
\]
Then,
\[
\psi_n^\star-\Psi(P_0)=
\frac{1}{V}\sum_{v=1}^V(P_{n,v}^1-P_0)D_{{\cal M}_{w},P_{n,v}^\star}+\frac{1}{V}\sum_{v=1}^V R_{{\cal M}_{w}}(P_{n,v}^\star,P_0)+o_P(n^{-1/2}).\]
{\bf Rate of convergence condition:}
Assume $1/V\sum_{v=1}^V R_{{\cal M}_{w}}(P_{n,v}^\star,P_0)=o_P(n^{-1/2})$. \newline
{\bf Weak Donsker class condition:}
Assume, conditional on the training sample $P_{n,v}$, $D_{{\cal M}_{w,n,v},P_{n,v}^\star}$ falls in a $P_0$-Donsker class (note that only the targeting step depends on $P_{n,v}^1$, making this a very weak condition).
Then,
\[
\psi_n^\star-\Psi(P_0)=O_P(n^{-1/2}).\]
{\bf Consistency of TMLE $P_{n,v}^\star$ to $P_0$:}
Suppose $\max_{v=1}^V\pl D_{{\cal M}_{w},P_{n,v}^\star}-D_{{\cal M}_{w},P_0}\pl_{P_0}=o_P(1)$.
Then,
\[
\psi_n^\star-\Psi(P_0)=
(P_n-P_0)D_{{\cal M}_{w},P_0}+o_P(n^{-1/2}).\]
Note that the right-hand side is a sample mean of mean zero independent random variables.\newline
Therefore, by the CLT, we have
\[
\sigma_{w,n}^{-1}(\psi_n^\star-\psi_0)\Rightarrow_d N(0,1),\]
where $\sigma_{w,n}^2=P_0 D_{{\cal M}_{w},P_0}^2$.\newline
{\bf Asymptotic linearity condition:}
Assume $\pl D_{{\cal M}_{w},P_0}-D_{{\cal M}_0,P_0}\pl_{P_0}=o_P(1)$. Under this stronger condition on ${\cal M}_{w}$, we have asymptotic linearity with a super-efficient influence curve: 
\[
\psi_n^\star-\Psi(P_0)=P_n D_{{\cal M}_0,P_0}+o_P(n^{-1/2}).\]
\end{theorem}
Note that the above theorem already provides asymptotically valid confidence intervals without relying on the asymptotic linearity condition. By also having the latter condition, the estimator is asymptotically linear with super-efficient influence curve $D_{{\cal M}_0,P_0}$ which comes with regularity with respect to the oracle model ${\cal M}_0$. 

\subsection*{Construction of data-adaptive working model through data-adaptive dimension reductions}
As in the definition of the adaptive CV-TMLE, let ${\cal M}_{w,n,v}\subset {\cal M}$ be a working model based on the training sample $P_{n,v}$, $v=1,\ldots,V$. However, here we want to consider a particular strategy for constructing such working models to which we can then apply the CV-TMLE, thereby obtaining a particular adaptive CV-TMLE described and analyzed in Appendix C. 

To demonstrate this proposal we consider a statistical model that corresponds to variation independent models of conditional densities. Specifically, let $O=(O_1,\ldots,O_J)$ and consider a  statistical model ${\cal M}$ that assumes $p(O)=\prod_{j=1}^J p_{O(j)\mid Pa(O_j)}(O(j)\mid Pa(O_j))$, but leaves all the conditional densities of $O(j)$, given its parents $Pa(O_j)$, unspecified, $j=1,\ldots,J$. For example, the likelihood might be factorized according to the ordering $p(O)=\prod_{j=1}^J p(O(j)\mid \bar{O}(j-1))$, where $\bar{O}(j-1)=(O(1),\ldots,O(j-1))$, and one might not make any assumptions, or one might assume that for some of the conditional densities $p(O(j)\mid \bar{O}(j-1))$ it is known to depend on $\bar{O}(j-1)$ through a dimension reduction $Pa(O(j))$. 
This defines then a statistical model ${\cal M}$ for $P_0$ only driven by conditional independence assumptions. 

{\bf Working model defined by estimated dimension reductions:}
For a given $j$, let $S_{n,j}=S_j(P_n)$ be a data dependent  dimension reduction in the sense that $S_{n,j}(Pa(O_j))$ is of (much) smaller dimension than $Pa(O_j)$. 
Consider the submodel ${\cal M}_{w}$ defined by assuming $p_{O(j)\mid Pa(O_j)}=p_{O(j)\mid S_{n,j}(Pa(O_j))}$, $j=1,\ldots,J$: in other words, this working model assumes that $O(j)$ is independent of $Pa(O_j)$, given the reduction $S_{n,j}(Pa(O_j))$, $j=1,\ldots,J$.
Clearly, we have ${\cal M}_{w}\subset{\cal M}$ by making stronger conditional independence assumptions than the ones defining ${\cal M}$.

The Kullback-Leibler projection of $p$ onto ${\cal M}_{w}$ would involve projecting each $p_j=p_{O(j)\mid Pa(O)j)}\in {\cal M}_{j}$ onto its smaller working model ${\cal M}_{j,n}$ as follows:
\[
\Pi(P_j\mid {\cal M}_{j,n})=\arg\min_{P_{1,j}\in {\cal M}_{j,n}}P L(P_{1,j}),\]
where $L(P)=-\log p$ is the log-likelihood loss. 
So this projection computes the MLE of $P_j$ over the working model under an infinite sample from $P$.
Therefore, we have a clear definition of $\Pi_P(P\mid {\cal M}_{w})$ and also a clear understanding of how one estimates this projection $\Pi_{P_0}(P_0\mid {\cal M}_{w})$ with an MLE $\tilde{P}_n=\arg\min_{P_1\in {\cal M}_{w}}P_n L(P_1)$ over ${\cal M}_{w}$, or, if ${\cal M}_{w}$ is too high dimensional, with a regularized MLE such as an HAL-MLE.  If $Pa(O(j))$ is very high dimensional, then an HAL-estimator of the conditional density $p_{0,j}$ is cumbersome, while, if ${\cal M}_{w}$ is indexed by a low dimensional unspecified function, then we can estimate this projection with a powerful HAL-MLE that is also computationally very feasible.

{\bf Obtaining dimension reduction through fitting the conditional density:}
Such scores $S_{n,j}$ could be learned by fitting the $j$-specific conditional density $p_{0,j}$ with a super-learner or with other state of the art machine learning algorithms. Typically, these estimators naturally imply a corresponding dimension reduction $S_{n,j}$, as we will show now.
 For example, if $O(j)$ is binary, then an estimator $p_{j,n}(1\mid Pa(O_j))$ of $p_{j,0}(1\mid Pa(O_j))$ implies the score $S_{n,j}(Pa(O_j))\equiv p_{j,n}(1\mid Pa(O_j))$. If $O(j)$ is discrete with $k_j$-values, then one could define $S_{n,j}(Pa(O_j))\equiv (p_{j,n}(k\mid Pa(O_j)): k=1,\dots,k_{j}-1)$ as a $ k_j-1$-dimensional score. Consider now the case that $p_n(y\mid x)$ is an estimator of a conditional density $p(y\mid x)$ of a continuous-valued $y$. One might then observe that $p_n(y\mid x)$ only depends on $x$ through a vector $S_n(x)$. 
 
{\bf Obtaining lower dimensional data-adaptive working model directly through a fit of the conditional density:}
Consider now the case that $Y$ is continuous and we want to determine a lower dimensional working model for $p(Y\mid X)$. Suppose our estimator $p_n(y\mid x)$ is fitted through a hazard fit $\lambda_n(y\mid x)$ that is of the form $\exp(\phi_n(y,x))$ for some low dimensional $\phi_n(y,x)$ (e.g., one dimensional, by taking the fit itself). This form $\lambda_n=\exp(\phi_n)$ does not necessarily suggest a score $S_n(x)$, but  one could use as submodel $\{ \lambda(y\mid x)=\exp(h(\phi_n(y,x))):h\}$ for  an arbitrary function $h$. This submodel is parameterized by a univariate function. 

{\bf Summary for obtaining low dimensional working models:}
Overall, we conclude that any kind of machine learning algorithm for fitting a conditional density, including highly aggressive super-learners,  will imply a low dimensional submodel for that conditional density, either a submodel that assumes conditional independence given a dimension reduction of the parent set or a submodel that parametrizes the conditional density in terms of a low dimensional function. 

{\bf General remarks:}
These type of working models ${\cal M}_w(P_n)$ could be highly data-adaptive making it important to carry out the adaptive CV-TMLE instead of the adaptive TMLE. Since the model ${\cal M}_{w,n,v}$ is much lower dimensional, one could use HAL-MLE or   a discrete super-learner based on various HAL-MLEs as initial estimator $P_{n,v}^0$ in the TMLE $P_{n,v}^\star$ targeting $\Psi_{{\cal M}_{w,n,v}}(P_0)$.  Therefore, once the data-adaptive working models ${\cal M}_{w,n,v}$ have been computed, the remaining part of the adaptive CV-TMLE is computationally feasible and can fully utilize a powerful theoretically grounded algorithm HAL with its strong rates of convergence (which then implies that  the rate of convergence condition and Donsker class condition of Theorem \ref{atmletheorem2} hold). To satisfy the conditions on ${\cal M}_w(P_n)$ w.r.t. approximating $P_0$, it will be important that the super-learners or other machine learning algorithms used to generate the working models for the conditional densities $p_{0,j}$ are converging to the true conditional densities at a rate $n^{-1/4}$.  Fortunately, these algorithms for learning the conditional densities have no restrictions and can be utilizing a large variety of machine learning approaches in the literature, including deep-learning, large language models and meta-HAL super learners \cite{wang_metahal_2023}. In this manner, we can utilize the full range of machine learning algorithms and computer power to obtain working models ${\cal M}_{w}(P_n)$ that approximate $P_0$ optimally.

\end{document}